\documentclass[11pt]{article}      
\usepackage[margin=1in]{geometry}  

\usepackage{libertine} 
\usepackage{inconsolata} 

\usepackage{microtype}
\usepackage[defaultlines=3,all]{nowidow}
\usepackage[export]{adjustbox}
\usepackage{xspace}
\usepackage{amsmath}
\usepackage{amsthm}
\usepackage{wrapfig}
\usepackage{thm-restate}
\usepackage{tikz}
\usepackage{tikz-cd}
\usepackage{quiver}

\usepackage{tcolorbox}

\newtcolorbox{crucialassumption}{%
  colback=gray!5,
  colframe=gray!50,
  boxrule=0.5pt,
  arc=0mm 
}

\usepackage[font=small]{caption}
\usepackage[hidelinks]{hyperref}
\usepackage[capitalise,nameinlink]{cleveref}
\usepackage{amsfonts}
\usepackage{amssymb}
\usepackage[textsize=tiny]{todonotes}
\usepackage{enumitem}
\usepackage{mathtools}
\usepackage{cases}
\usepackage[only,llbracket,rrbracket]{stmaryrd}

\usepackage{array}
\usepackage{multirow}
\usepackage{multicol}
\usepackage{mathrsfs} 
\newtheorem{theorem}{Theorem}[section]
\newtheorem{corollary}[theorem]{Corollary}
\newtheorem{lemma}[theorem]{Lemma}
\newtheorem{claim}[theorem]{Claim}

\newtheorem{conjecture}[theorem]{Conjecture}
\theoremstyle{definition}
\newtheorem{definition}[theorem]{Definition}
\newtheorem{observation}[theorem]{Observation}

\newtheorem{example}[theorem]{Example}
\newtheorem{remark}[theorem]{Remark}

\newtheorem{fact}[theorem]{Fact}

\usetikzlibrary{arrows.meta}

\usepackage{enumitem}
\setlist[itemize]{topsep=4pt,itemsep=3pt,parsep=0pt} 
\setlist[enumerate]{topsep=4pt,itemsep=3pt,parsep=0pt} 

\Crefname{figure}{Figure}{Figures}

\renewcommand{\cal}{\mathcal}

\renewcommand{\int}{\mathrm{int}}

\newcommand{\NN}[0]{\mathrm{\mathbb{N}}}
\newcommand{\N}[0]{\mathrm{\mathbb{N}}}
\newcommand{\Oof}{\mathcal{O}}

\newcommand{\dist}{\textnormal{dist}}

\renewcommand{\phi}{\varphi}

\newcommand{\Start}{\mathrm{start}}
\newcommand{\End}{\mathrm{end}}

\newcommand{\LL}{\mathcal{L}}

\newcommand{\DD}{\mathscr{D}}

\newcommand{\XX}{\mathcal{X}}
\newcommand{\Dd}{\mathscr{D}}

\newcommand{\Pp}{\mathcal{P}}
\newcommand{\Cc}{\mathscr{C}}

\newcommand{\CC}{\mathscr{C}}

\newcommand{\MM}{\mathcal{M}}
\newcommand{\PP}{\mathcal{P}}
\newcommand{\QQ}{\mathcal{Q}}
\newcommand{\RR}{\mathcal{R}}

\renewcommand{\le}{\leqslant}
\renewcommand{\leq}{\le}
\renewcommand{\ge}{\geqslant}
\renewcommand{\geq}{\ge}

\newcommand{\trans}[1]{\mathsf{#1}}
\newcommand{\expos}{EP}

\usepackage{calc}
\newcommand{\textover}[3][l]{%
 \makebox[\widthof{#3}][#1]{#2}%
 }

\newcommand{\needsselfloops}{$\circlearrowleft$}

\newenvironment{claimproof}[1][\proofname]{%
  \begin{proof}[#1]%
}{%
  \end{proof}%
}

\begin{document}

\newcommand{\funding}{NM received funding from the European Research Council (ERC) with grant agreement No.\ 101126229 -- {\sc buka}.
}

\title{Existential Positive Transductions of Sparse Graphs\thanks{\funding}}
\date{}
\author{
  Nikolas M\"ahlmann \\
  \small{University of Warsaw} \\
  \small{\texttt{maehlmann@mimuw.edu.pl}}
  \and
  Sebastian Siebertz\\
  \small{University of Bremen} \\
  \small{\texttt{siebertz@uni-bremen.de}}
}
\maketitle

\begin{abstract}
    Monadic stability generalizes many tameness notions from structural graph theory such as planarity, bounded degree, bounded tree-width, and nowhere density.
    The \emph{sparsification conjecture} predicts that the (possibly dense) monadically stable graph classes are exactly those that can be logically encoded by first-order (FO) transductions in the (always sparse) nowhere dense classes.
    So far this conjecture has been verified for several special cases, such as for classes of bounded shrub-depth, and for the monadically stable fragments of bounded (linear) clique-width, twin-width, and merge-width. 

    In this work we propose the \emph{existential positive sparsification conjecture}, predicting that the more restricted co-matching-free, monadically stable classes are exactly those that can be transduced from nowhere dense classes using only existential positive FO formulas.
    While the general conjecture remains open, we verify its truth for all known special cases of the original conjecture.
    Even stronger, we find the sparse preimages as subgraphs of the dense input graphs. 

    As a key ingredient, we introduce a new combinatorial operation, called \emph{subflip}, that arises as the natural co-matching-free analog of the flip operation, which is a central tool in the characterization of monadic stability. 
    Using subflips, we characterize the co-matching-free fragment of monadic stability by appropriate strengthenings of the known flip-flatness and flipper game characterizations for monadic stability.
    
    In an attempt to generalize our results to the more expressive MSO logic, we discover (rediscover?) that on relational structures (existential) positive MSO has the same expressive power as (existential) positive FO.
\end{abstract}

\begin{picture}(0,0)
\put(422,-220)
{\hbox{\includegraphics[width=40px]{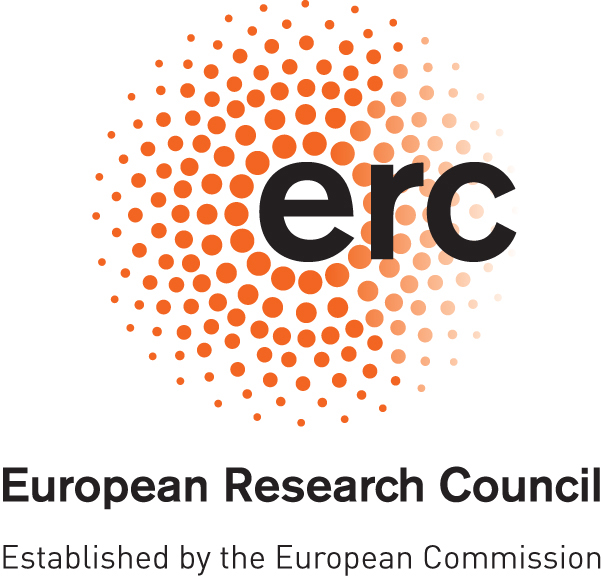}}}
\put(412,-280)
{\hbox{\includegraphics[width=60px]{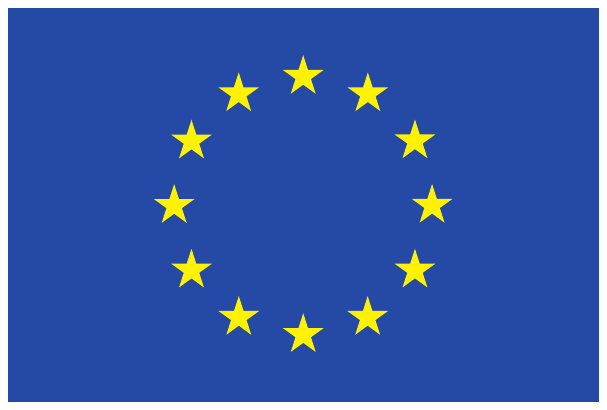}}}
\end{picture}

\newpage

\tableofcontents
\newpage

\section{Introduction}

First-order transductions have emerged as a fundamental tool in structural graph theory and finite model theory, providing a convenient way to encode one class of structures inside another using logic. 
Informally, a \emph{transduction} $\trans T_\phi$ is an operation specified by a binary first-order formula $\phi(x,y)$ over the signature of colored graphs.
It maps an input graph $G$ to the set $\trans T_\phi(G)$ containing every graph~$H$ that can be obtained through the following three-step process. (See \cref{fig:transduction} for an example.)

\newcommand{\GLeft}[1]{\textover{#1}{$G'$}}

\begin{enumerate}
    \item \GLeft{$G$} $\rightarrow$ \GLeft{$G'$} : 
    color $G$ with vertex colors appearing in $\phi$.
    
    \item \GLeft{$G'$} $\rightarrow$ \GLeft{$H'$} : 
    let $\phi$ define a new edge relation: $uv \in E(H') \Leftrightarrow G' \models \phi(u,v)$.    

    \item \GLeft{$H'$} $\rightarrow$ \GLeft{$H$} :
    choose $H$ to be any induced subgraph of $H'$ and remove the vertex colors.
\end{enumerate}

\begin{figure}[htbp]
    \centering
    \includegraphics[width = \textwidth]{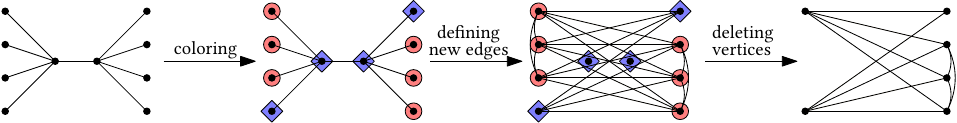}
    \caption{
    Example of a transduction. On the very left: the graph $G$. On the very right: a graph $H \in \trans T_\phi(G)$ for the formula $\phi(x,y) = (\dist(x,y) = 3) \vee (\textnormal{Red}(x) \wedge \textnormal{Red}(y))$.}
    \label{fig:transduction}
\end{figure}

Further graph transformations that can be expressed by transductions include the complement and the square operation, specified by the formulas
$\neg E(x,y)$, and $E(x,y) \vee \exists z\, \big(E(x,z) \wedge E(z,y)\big)$, respectively.
This motivates the study of properties of graph classes that are preserved under transductions. 
A popular example of such a \emph{transduction-closed} class property is bounded clique-width: for every graph class $\CC$ whose clique-width is bounded by a constant and every transduction~$\trans T$, the clique-width of the  class $\trans{T}(\CC) := 
\bigcup_{G\in\CC} \trans T(G)$ is again bounded by a (possibly larger) constant.
For every class $\DD \subseteq \trans{T}(\CC)$, we say that $\DD$ is a \emph{transduction} of $\CC$, and that $\CC$ \emph{transduces} $\DD$.
Further examples of transduction-closed class properties include bounded shrub-depth, bounded linear clique-width, bounded twin-width, and bounded merge-width.
Subsuming all previous examples, the most general non-trivial, transduction-closed class property is \emph{monadic dependence}: a class is \emph{monadically dependent} if it does not transduce the class of all graphs.
In this paper we study the relationship between three fragments of monadic dependence: 1.\ the \emph{biclique-free} (\emph{weakly sparse}) fragment,  2.\ the \emph{semi-ladder-free} fragment, and 3.\ the \emph{half-graph-free} (\emph{stable}) fragment.

\begin{restatable}{definition}{defSemiinduced}
    We call a bipartite graph $G$ on vertices $a_1 ,\ldots, a_n$ and $b_1, \ldots, b_n$
    \begin{itemize}
        \item the \emph{biclique} $K_{n,n}$ of order $n$, if $a_ib_j \in E(G)$ for all $i,j \in [n]$,
        \item the \emph{half-graph} of order $n$, if $a_ib_j \in E(G) \Leftrightarrow i \leq j $ for all $i,j \in [n]$, and
        \item the \emph{co-matching} of order $n$, if $a_ib_j \in E(G) \Leftrightarrow i \neq j $ for all $i,j \in [n]$.
    \end{itemize}
    A bipartite graph $H$ is a \emph{semi-induced subgraph} of a graph $G$, if there are $A, B \subseteq V(G)$ such that the bipartite subgraph of $G$ induced between the sides $A$ and $B$ is isomorphic to $H$.
    A graph class is \emph{biclique-free} / \emph{half-graph-free} / \emph{co-matching-free} if it does not contain arbitrarily large bicliques / half-graphs / co-matchings as semi-induced subgraphs.
    The class is \emph{semi-ladder-free} if it is both half-graph-free and co-matching-free.
    See \Cref{fig:hierarchy} for an illustration.
\end{restatable}

\begin{figure}[htbp]
    \centering
    \includegraphics[scale = 1]{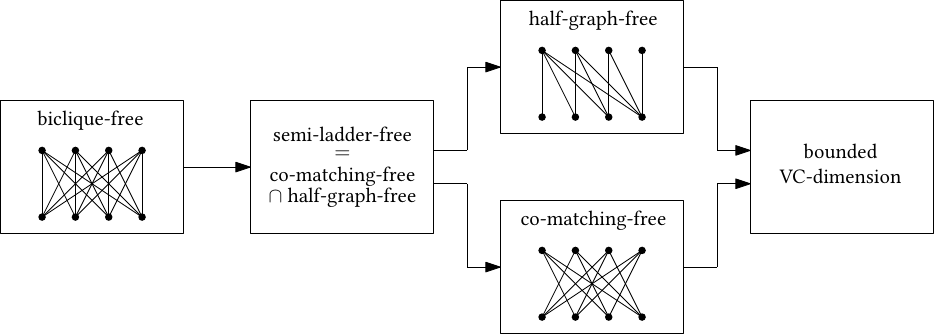}
    \caption{A hierarchy of properties of graph classes.
    An arrow $P_1 \rightarrow P_2$ between two properties means that every graph class that has property $P_1$ also has property $P_2$.}
    \label{fig:hierarchy}
\end{figure}

The \emph{biclique-free} fragment of monadic dependence (also called the \emph{weakly sparse} fragment) consists exactly of the \emph{nowhere dense} classes.
The notion of nowhere density is the main object of study in sparsity theory initiated by Nešetřil and Ossona de Mendez~\cite{nevsetvril2012sparsity}.
It generalizes many tameness notions for sparse graphs such as planarity, bounded tree-width, excluded minors, bounded degree, and bounded expansion. 
The equivalence of the combinatorial notion of nowhere density and the logical notion of monadic dependence in the biclique-free setting is one of the key results that initially connected sparsity theory and model theory.
It follows from the results of
Podewski and Ziegler~\cite{stable_graphs}, Adler and Adler~\cite{adler2014interpreting}, and Dvo\v{r}\'ak~\cite{dvovrak2018induced}. 

Generalizing nowhere density, the \emph{half-graph-free} fragment of monadic dependence consists exactly of the \emph{monadically stable} classes~\cite{nevsetvril2021rankwidth}, which are those that do not transduce the class of all half-graphs~\cite{baldwin1985second}.
While nowhere density is limited to sparse classes, monadic stability also includes dense classes like map graphs and every class of bounded shrub-depth.
Monadic stability is transduction-closed and hence, every transduction of a nowhere dense class is monadically stable.
The \emph{sparsification conjecture} predicts an even stronger link between the biclique-free and the half-graph-free fragments of monadic dependence, asserting that the reverse should hold as well.

\begin{conjecture}[see \cite{POM21}]\label{conjecture:stable-general}
    For every graph class $\CC$ the following are equivalent.
    \begin{enumerate}
        \item $\CC$ is monadically stable.
        \item $\CC$ is a transduction of a nowhere dense class.
    \end{enumerate}
\end{conjecture}

So far the sparsification conjecture has been confirmed for all half-graph-free classes $\CC$ that have bounded shrub-depth~\cite{shrubdepth,shrubdepth-journal},  linear clique-width~\cite{nevsetvril2020linear},  clique-width~\cite{nevsetvril2021rankwidth}, twin-width~\cite{gajarsky2022stable}, or merge-width~\cite{dreier2026efficient}.
In each of these cases actually something stronger holds. 
For example, every half-graph free class of bounded clique-width can be transduced from a class that is not only nowhere dense, but even has bounded tree-width~\cite{nevsetvril2021rankwidth} (which is the same as having both bounded clique-width and being biclique-free~\cite{courcelle2000upper, gurski2000tree}).
This suggests the following refinement of \Cref{conjecture:stable-general}.

\begin{conjecture}\label{conjecture:stable-refined}
    Let $\Pp$ be a non-trivial and transduction-closed class property. 
    For every graph class $\CC$, the following are equivalent. 
    \begin{enumerate}
        \item $\Cc$ is a half-graph-free class with $\Pp$.
        \item $\Cc$ is a transduction of a biclique-free class with $\Pp$. 
    \end{enumerate}
\end{conjecture}

Note that this conjecture implies the sparsification conjecture by choosing $\PP$ to be monadic dependence.
The solved cases of the sparsification conjecture can now be stated as follows.

\begin{theorem}\label{thm:general-conjecture-sbe}
    \cref{conjecture:stable-refined} is true for the class property \emph{bounded $\PP$} for all
    \begin{center}
        $\PP \in \{\text{shrub-depth, linear clique-width, clique-width, twin-width, merge-width}\}$.
    \end{center}
    In particular, for every half-graph-free class $\CC$, all of the following hold.
    \begin{enumerate}
        \item $\CC$ has bounded shrub-depth if and only if 
        
        $\CC$ is a transduction of a class of bounded tree-depth~\cite{shrubdepth,shrubdepth-journal}.

        \item $\CC$ has bounded linear clique-width if and only if 
        
        $\CC$ is a transduction of a class of bounded path-width~\cite{nevsetvril2020linear}.

        \item $\CC$ has bounded clique-width if and only if 
        
        $\CC$ is a transduction of a class of bounded tree-width~\cite{nevsetvril2021rankwidth}.

        \item $\CC$ has bounded twin-width if and only if 
        
        $\CC$ is a transduction of a class of bounded sparse twin-width~\cite{gajarsky2022stable}.

        \item $\CC$ has bounded merge-width if and only if 
        
        $\CC$ is a transduction of a class of bounded expansion~\cite{dreier2026efficient}.
    \end{enumerate}
\end{theorem}

The ``in particular'' part of the theorem uses the following fact. 

\begin{theorem}\label{lem:sparse-notions}
Every biclique-free class $\Cc$ with 
\begin{enumerate}
    \item bounded shrub-depth has bounded tree-depth~\cite{shrubdepth,shrubdepth-journal},
    \item bounded linear clique-width has bounded path-width~\cite{gurski2000tree},
    \item bounded clique-width has bounded tree-width~\cite{courcelle2000upper, gurski2000tree},
    \item bounded twin-width has bounded sparse twin-width~\cite{twwII},
    \item bounded merge-width has bounded expansion~\cite{merge-width}.
\end{enumerate} 
\end{theorem}


\subsection*{Contribution 1: Existential positive sparsification}
The main object of our study is the \emph{semi-ladder-free} fragment of monadic dependence, which strictly subsumes the nowhere dense classes and is strictly contained in the monadically stable classes.
A class is semi-ladder-free if and only if it is half-graph-free and co-matching-free~\cite{fabianski2019}.
Prominent examples of semi-ladder-free monadically dependent classes are the class of map graphs and fixed powers of nowhere dense classes~\cite{fabianski2019}. 
Analogously to the sparsification conjecture, we conjecture that the semi-ladder-free monadically dependent classes are exactly the \emph{existential positive} transductions of nowhere dense classes.
More formally, a formula is \emph{positive} if it contains no negation symbol~$(\neg)$, and in particular no non-equality~$(\neq)$. A positive formula is moreover \emph{existential} if it contains no universal quantifier $(\forall)$. In this case we call it an \emph{\expos-formula}.
We call a transduction an \emph{\expos-transduction} if it is specified by an \expos-formula.

It turns out that when studying \expos-transductions it is crucial that all vertices have self-loops (i.e., every vertex is adjacent to itself). 
Intuitively, the reason for this is that otherwise the information that two vertices are connected by an edge implicitly carries the non-positive information that these vertices are non-equal. 
We will discuss this phenomenon in more detail in the next subsection and also in \Cref{sec:self-loops}. 
In the following, a \emph{reflexive graph} is a graph in which every vertex has a self-loop.
For additional clarity, we mark statements requiring this assumption with a \needsselfloops\ symbol.
On the other hand, if we do not mention reflexivity explicitly, this means we allow self-loops in our graphs, but do not require them. Sometimes we use the term \emph{partially reflexive graph}, to stress that we are working in this more relaxed setting.

\bigskip
We propose the following \emph{existential positive sparsification conjecture}.

\begin{conjecture}[\needsselfloops]\label{conjecture:positive-general}
    For every class $\CC$ of reflexive graphs, the following are equivalent.
    \begin{enumerate}
        \item $\CC$ is semi-ladder-free and monadically dependent.
        \item $\CC$ is an \expos-transduction of a nowhere dense class of reflexive graphs.
    \end{enumerate}
\end{conjecture}

Note that according to~\cite{nevsetvril2021rankwidth} and~\cite{fabianski2019} the first condition that $\CC$ is semi-ladder-free and monadically dependent is equivalent to the condition that $\Cc$ is co-matching-free and monadically stable. 

\medskip
Again, the statement can be refined to subproperties of monadic dependence.

\begin{conjecture}[\needsselfloops]\label{conjecture:positive-refined}
    Let $\Pp$ be a non-trivial and transduction-closed class property. For every class~$\CC$ of reflexive graphs, the following are equivalent. 
    \begin{enumerate}
        \item $\Cc$ is a semi-ladder-free class with $\Pp$.
        \item $\Cc$ is an \expos-transduction of a biclique-free class of reflexive graphs with $\PP$. 
    \end{enumerate}
\end{conjecture}

As our first main result, we establish the implication (2.\ $\Rightarrow$ 1.) of \Cref{conjecture:positive-general}. For this direction no assumptions on self-loops are required.

\begin{restatable}{theorem}{thmCoMatchingFreenessPreserved}\label{thm:co-matching-freeness-preserved}
    Let $\CC$ be a nowhere dense class of partially reflexive graphs. 
    Every \expos-transduction of $\CC$ is semi-ladder-free and monadically dependent.
\end{restatable}

As a corollary we obtain that also the implication (2.\ $\Rightarrow$ 1.) of \Cref{conjecture:positive-refined} is true.

\begin{corollary}\label{cor:preservation}
    The implication (2.\ $\Rightarrow$ 1.) of \Cref{conjecture:positive-refined} is true.
\end{corollary}
\begin{proof}
Let $\Cc$ be an \expos-transduction of a biclique-free class $\Dd$ with some non-trivial, transduction-closed class property $\Pp$. 
$\DD$ must be monadically dependent, because monadic dependence is the most general non-trivial, transduction-closed class property.
Then by biclique-freeness, $\DD$ is nowhere dense. By \Cref{thm:co-matching-freeness-preserved}, $\CC$ is semi-ladder-free. As $\PP$ is transduction-closed, $\CC$ still has $\PP$.
\end{proof}

Consider the following algorithmic application of our structural result. 
It was proved in~\cite{dreier2022combinatorial} that the \textsc{Dominating Set} and \textsc{Independent Set} problem admit a polynomial kernel (with respect to solution size~$k$) on every co-matching-free monadically stable class of graphs. 
We conclude that polynomial kernels for these problems exist on all \expos-transductions of nowhere dense classes. 

\medskip

With the implication (2.\ $\Rightarrow$ 1.) of the existential positive sparsification conjecture settled, we continue to prove the full equivalence for all special cases for which the original sparsification conjecture is known to be~true.

\begin{restatable}[\needsselfloops]{theorem}{thmSparsification}\label{thm:sparsification}
    \cref{conjecture:positive-refined} is true for the class property \emph{bounded $\PP$} for all
    \begin{center}
        $\PP \in \{\text{shrub-depth, linear clique-width, clique-width, twin-width, merge-width}\}$.
    \end{center}
    In particular, for every semi-ladder-free class $\CC$ of reflexive graphs there are \expos-transductions $\trans{Sparsify}$ and $\trans{Recover}$ such that every graph $G\in\CC$ contains a reflexive subgraph $G^*$ such that 
    \begin{center}
        $G^* \in \trans{Sparsify}(G)$ and $G \in \trans{Recover}(G^*)$,
    \end{center}
    and the class $\DD := \{G^* : G\in \CC\}$ satisfies all of the following.

    \begin{enumerate}
        \item If $\CC$ has bounded shrub-depth, then $\DD$ has bounded tree-depth.

        \item If $\CC$ has bounded linear clique-width, then $\DD$ has bounded path-width.

        \item If $\CC$ has bounded clique-width, then $\DD$ has bounded tree-width.

        \item If $\CC$ has bounded twin-width, then $\DD$ has bounded sparse twin-width.

        \item If $\CC$ has bounded merge-width, then $\DD$ has bounded expansion.
    \end{enumerate}
\end{restatable}

Note that in the above theorem the sparsified graphs are obtained as \emph{subgraphs of the original input graphs}. 
This is in contrast to the known cases of the general sparsification conjecture (\cref{thm:general-conjecture-sbe}), where one merely obtains the existence of some graph class with the desired structural properties that is transduction-equivalent to~$\Cc$, without any guarantee that the witnessing graphs are subgraphs of the original ones. 
Thus, for all known cases of the conjecture, the theorem establishes a substantially stronger, canonical form of sparsification, in which the structure of $G$ is preserved and refined rather than abstractly re-encoded. 

An overview of the solved cases of the (existential positive) sparsification conjecture and of the relationships between the appearing class properties is given in \Cref{fig:universe}.

\begin{figure}[htbp]
    \centering
    \includegraphics[width = \textwidth]{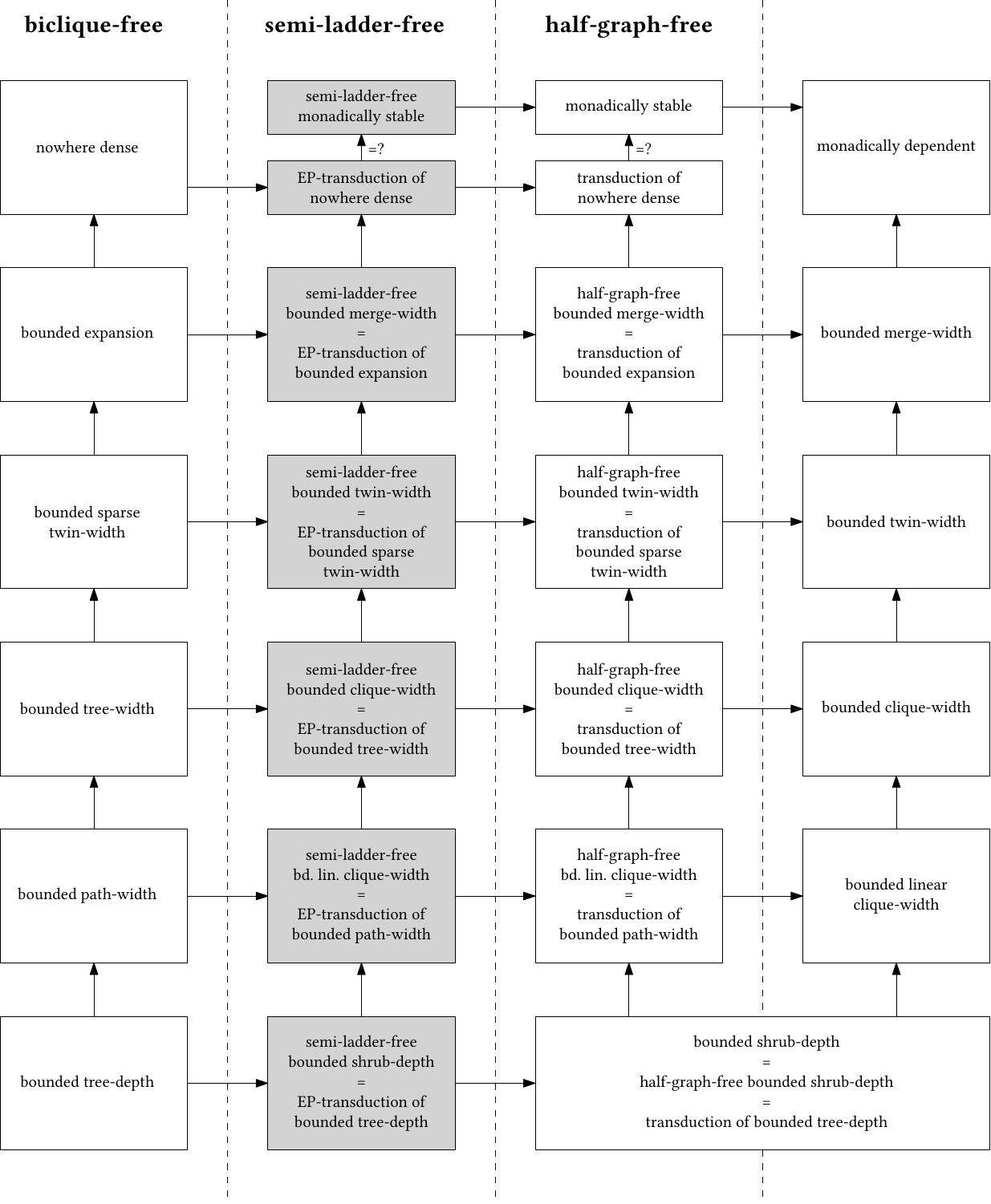}
    \caption{A hierarchy of properties of graph classes.
    An arrow $P_1 \rightarrow P_2$ between two properties means that every graph class that has property $P_1$ also has property $P_2$. The highlighted semi-ladder-free column represents our results on the existential positive sparsification conjecture.}
    \label{fig:universe}
\end{figure}

\subsection*{Contribution 2: The monadic non-equality property}
We take some time to distill the model theoretic core, that helps to explain both the connection between semi-ladder-freeness and \expos-transductions, and the need to work with reflexive graphs.

Let $\phi(\bar x, \bar y)$ be an FO formula over the signature of colored graphs, where $\bar x$ and $\bar y$ are two tuple variables of the same length~$t$.
We say that $\phi$ has the \emph{monadic order-property} on a graph class~$\CC$ if for every $\ell \in \N$ there is a graph coloring $G_\ell^+$ of a graph $G_\ell \in \CC$ and tuples $\bar a_1, \ldots, \bar a_\ell \in V(G)^t$ such that for all $i,j\in [\ell]$
\begin{equation}\label{eq:op}
    G^+_\ell \models \phi(\bar a_i, \bar a_j) \iff i \leq j. \tag{$*$}
\end{equation}

While we have previously defined monadic stability via transductions, the original (and equivalent) definition by Baldwin and Shelah~\cite{baldwin1985second} states that a graph class $\CC$ is monadically stable if no formula has the monadic order-property on $\CC$.
It is now natural to define the \emph{monadic non-equality-property} by replacing~\eqref{eq:op} in the definition of the monadic order property with
\[
    G^+_\ell \models \phi(\bar a_i, \bar a_j) \iff i \neq j.
\]

The simple formula $\phi(x,y) = \neg (x = y)$ has the monadic non-equality-property on every graph class that contains arbitrarily large graphs.
Hence, the monadic non-equality-property is only interesting for positive formulas.
We show that for reflexive graphs, the monadic non-equality-property for \expos-formulas characterizes the semi-ladder-free fragment of monadic dependence.

\begin{restatable}[\needsselfloops]{theorem}{thmSemiLadderFreeNEP}\label{thm:nep-simple}
    For every class $\CC$ of reflexive graphs the following are equivalent.
    \begin{enumerate}
        \item $\CC$ is semi-ladder-free and monadically dependent.
        \item No \expos-formula has the monadic non-equality-property on $\CC$.
    \end{enumerate}
\end{restatable}

Above, assuming reflexiveness is crucial, and it is illustrative to see the difference to the general partially reflexive case, where we prove the following.

\begin{restatable}{theorem}{thmNEP}\label{thm:nep}
    For every class $\CC$ of partially reflexive graphs the following are equivalent.
    \begin{enumerate}
        \item $\CC$ is semi-ladder-free, irreflexive-clique-free, and monadically dependent.
        \item No \expos-formula has the monadic non-equality-property on $\CC$.
    \end{enumerate}
\end{restatable}

Here, an \emph{irreflexive} graph is a graph, in which no vertex has a self-loop.
We call a graph class \emph{irreflexive-clique-free} if it does not contain all irreflexive cliques as induced subgraphs.
Clearly, on the class of irreflexive cliques, the edge relation $E(x,y)$ has the monadic non-equality-property.
This hints at the reason why in the setting of positive logic, reflexive graphs behave nicer than irreflexive ones:
in irreflexive graphs, the positive information of being connected by an edge carries the negative information of being non-equal.
As mentioned before, we further discuss reflexivity in \cref{sec:self-loops}.

\subsection*{Contribution 3: Subflips}

As a foundation for the logical results presented in Contributions 1 and 2, we initiate the systematic combinatorial study of the semi-ladder-free fragment of monadic dependence. We obtain multiple characterizations through an operation, which we call \emph{subflip}, which we believe to be of independent interest.
Our starting point are the more general monadically stable classes (i.e., the half-graph-free fragment of monadic dependence).
Only recently, fixed-parameter tractability of the FO model checking problem was established for all monadically stable classes~\cite{dreier2024stablemc,dreier2023ssmc}.
Key to these algorithmic results was the emergence of a multitude of characterizations of monadic stability that are combinatorial rather than model theoretic in nature.
At the heart of these combinatorial characterizations lies the \emph{flip} operation.
Given a graph $G$, a \emph{$k$-flip} of $G$ is obtained by choosing a partition $\PP$ of $V(G)$ of size at most $k$ and complementing the edge relation between arbitrary pairs of parts from $\PP$ (possibly complementing parts with themselves, but preserving self-loops).
See \Cref{fig:flips} for an example.

\begin{figure}[htbp]
    \centering
    \includegraphics[scale = 0.7]{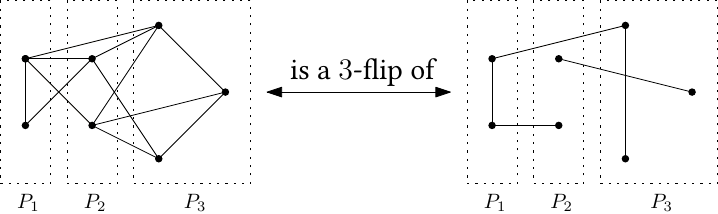}
    \caption{Two graphs that are $3$-flips of each other.
    This is witnessed by the partition $\{ P_1,P_2,P_3 \}$ where the pairs $(P_1,P_2)$, $(P_2,P_3)$, and $(P_3,P_3)$ were flipped. 
    }
    \label{fig:flips}
\end{figure}

The goal is usually to separate sets of vertices by putting them at high distance, using a $k$-flip for some small~$k$.
Exemplary for this is the following characterization of monadic stability by a property called \emph{flip-flatness}, which is one of the main tools for the algorithmic treatment of monadically stable classes.

\begin{definition}\label{def:flip-flatness}
    A graph class $\CC$ is \emph{flip-flat}
    if for every $r\in\N$ there is a function
    $M_r : \N \to \N$ and a constant $k_r \in \N$, such that for all $m \in \N$, $G \in \CC$, and $W \subseteq V(G)$ with
    $|W| \ge M_r(m)$,  there is a subset $A \subseteq W$ with $|A| \ge m$
    and a $k_r$-flip $H$ of $G$, such that for every two distinct vertices $u,v \in A$
    \[
    \dist_{H}(u,v) \geq r.    
    \]
\end{definition}

\begin{theorem}[\cite{dreier2022indiscernibles}]
    A graph class is monadically stable if and only if it is flip-flat.
\end{theorem}

Intuitively, flip-flatness states that in every huge set $W$, we find a still large subset $A$ whose vertices are pairwise at high distance (``well-separated'') after performing a small flip.

In this work we consider a restriction of the flip operation that we call \emph{subflip}. A graph $H$ is a \emph{$k$-subflip} of a graph $G$ if it is both a $k$-flip and a subgraph of $G$.
This means subflips are only allowed to remove edges, so we can only complement the edge relation between two parts $P$ and $Q$ of the flip partition $\PP$ of $V(G)$ if
\begin{itemize}
    \item $P$ and $Q$ are the same part and the graph induced on $P = Q$ is a clique, or
    \item $P$ and $Q$ are distinct parts and the graph semi-induced between $P$ and $Q$ is a biclique.
\end{itemize}

Being more restricted, subflips have some advantages over flips.
To give an example, the next lemma from~\cite{flip-separability} states that multiple flips of the same graph can be aggregated into a single flip that roughly achieves the same separation. 
This aggregation of flips is lossy and highly non-trivial to prove.

\begin{lemma}[{\cite[Lem.\ 8]{flip-separability}}]
    Let $G$ be a graph and $H_1, \ldots, H_m$ be $k$-flips of $G$. 
    There is an $\Oof(k^m \cdot 2^{k^m})$-flip~$H$ of $G$ such that for all $u,v \in V(G)$ and $i\in [m]$
    \[
        6 \cdot \dist_H(u,v) \geq \dist_{H_i}(u,v).
    \]
\end{lemma}

In contrast, aggregation of subflips is simple and lossless by working with the common refinement of the two flip partitions. (We later prove a strengthening of this observation in \Cref{lem:refine}.)

\begin{observation}
    Let $G$ be a graph and $H_1,\ldots, H_m$ be a $k$-subflips of $G$. There is a $k^m$-subflip~$H$ of $G$ that is a subgraph of all the $H_i$.
\end{observation}

This tameness comes at the cost of subflips being less powerful than unrestricted flips. 
A $2$-flip can turn any co-matching into a matching, by flipping between the two parts of the bipartition.
This flip achieves very strong separation: after the flip, every vertex has a single neighbor and distance $\infty$ to all other vertices. This is not possible to achieve with a small subflip.
However, it turns out that co-matchings are in some sense the only graphs were subflips perform worse than flips.
As one of our main results, we show that in co-matching-free classes, flips can always be approximated using subflips.

\begin{theorem}
    Let $G$ be a graph that contains no semi-induced co-matching of size $t$ and let $H$ be a $k$-flip of~$G$.
    There exists a $kt^k$-subflip $H'$ of~$G$ such that for all $u,v \in V(G)$
    \[
        3t \cdot \dist_{H'}(u,v) \geq \dist_{H}(u,v).
    \]
\end{theorem}

Leveraging the above theorem, we obtain multiple combinatorial characterizations of semi-ladder-free monadically dependent classes, by replacing flips with subflips in known characterizations of monadic stability.
We need to make no assumptions about self-loops here.

\begin{theorem}\label{thm:helly-characterization}
For every graph class $\Cc$ the following are equivalent.
\begin{enumerate}
    \item $\Cc$ is semi-ladder-free and monadically dependent. 
    \item $\Cc$ is subflip-flat.
    \item Subflipper has a winning strategy for the subflipper-game on $\Cc$.
    \item $\CC$ is semi-ladder-free and for every $r,k \in \N$ both of the following hold:
    \begin{itemize}
        \item There is a star $r$-crossing that is not a $k$-subflip of any induced subgraph from $\CC$.
        \item There is a clique $r$-crossing that is not a $k$-subflip of any induced subgraph from $\CC$.
    \end{itemize}
\end{enumerate}
\end{theorem}
Naturally, the definition of \emph{subflip-flatness} is obtained by replacing flips with subflips in the definition of flip-flatness (\Cref{def:flip-flatness}). The \emph{subflipper-game} is similarly derived from the \emph{flipper-game}, which is a game theoretic characterization of monadic stability~\cite{flipper-game}.
The last item restricts a characterization of monadic stability by forbidden induced subgraphs~\cite{dreier2024stablemc,dreier2024flipbreakability} to the semi-ladder-free setting. 
We provide the formal definitions in the respective sections. 
\subsection*{Contribution 4: Collapse of (existential) positive MSO}

Finally, in exploring the possibility of generalizing our results to the more expressive monadic second-order logic (MSO), we uncover that on relational structures (existential) positive MSO collapses to (existential) positive FO.

\begin{theorem}
    Every (existential) positive MSO formula $\phi(\bar x, \bar X)$ is equivalent to an (existential) positive FO formula.
\end{theorem}

This theorem is not hard to prove.
However, to the best of our knowledge, it was not known before in the literature, and we find it worth stating.

\subsection*{Organization of the paper}
We give preliminaries in \Cref{sec:preliminaries}.
We further discuss self-loops in \Cref{sec:self-loops}.
In \Cref{sec:subflips,sec:sub-flatness,sec:subflipper-rank,sec:forbidden-subflips} we introduce subflips and prove the combinatorial characterizations of the semi-ladder-free fragment of monadic stability, that amount to \Cref{thm:helly-characterization}.
In \Cref{sec:normal-form,sec:nep}, we characterize the monadic non-equality-property.
In \Cref{sec:sparsification-into-subgraphs,sec:sbe} we show how to find sparse preimages for \expos-transductions as subgraphs (\Cref{thm:sparsification}).
In \Cref{sec:positive-mso}, we show how (existential) positive MSO collapses to (existential) positive FO.
\section{Preliminaries}\label{sec:preliminaries}

We use standard notation from graph theory and model theory
and refer to~\cite{Diestel} and~\cite{Hodges} for more background. 
We write $[m]$ for the set of integers $\{1,\ldots,m\}$.

\subsection*{Colored graphs}

In the context of logic we treat a \emph{graph} as a finite structure whose signature consists of a single binary relation symbol~$E$, interpreted as a symmetric edge relation. 
Given a graph $G$, we write~$V(G)$ for its vertex set and $E(G)$ for its edge set. 
We also allow self-loops, i.e., we may have $E(v,v)$ for vertices~$v$ (see \Cref{sec:self-loops} for a discussion of this assumption). 
A graph is called \emph{reflexive} if every vertex has a self-loop and \emph{irreflexive} if no vertex has a self-loop.
A graph $H$ is a \emph{subgraph} of $G$ if $V(H)\subseteq V(G)$ and
$E(H)\subseteq E(G)$. 
For graphs~$G,H$ (whose vertex sets may overlap) we write $G\cup H$ for the graph with vertex set $V(G)\cup V(H)$ and edge set $E(G)\cup E(H)$. We sometimes write $G \uplus H$ if the vertex sets of $G$ and $H$ are disjoint.
We recall the following definition from the introduction.

\defSemiinduced*

The \emph{co-matching-index} of a graph $G$ is the largest number $t$ such that $G$ contains a semi-induced co-matching of order $t$.

A \emph{monadic extension} of the graph signature is obtained by adding (finitely many) unary predicate symbols. 
We usually do not distinguish between a relation symbol and its interpretation. 
We refer to unary predicates as \emph{colors}.
A \emph{colored graph} is a graph expanded by a (finite) set of
colors. 

If $\Cc$ is a class of graphs and $\sigma^+$ is a monadic extension of the graph signature, we write $\Cc[\sigma^+]$ for the class of all $\sigma^+$-expansions of graphs from~$\Cc$.

\subsection*{First-order logic}

First-order logic over a signature $\sigma$ is defined in 
the usual way.
We usually write $\bar x$ for tuples $(x_1,\ldots, x_k)$
of variables, $\bar a$ for tuples $(a_1,\ldots, a_\ell)$ of elements and leave it to the context to determine the length of the tuple.
A formula is \emph{existential} if it contains no universal quantifiers and negation is applied
only to quantifier-free subformulas
A formula is \emph{positive} if it does not use negations (in particular it cannot use $\neq$). 
We call an existential positive formula an \emph{\expos-formula}. 

As a shorthand, we write $\dist(x,y) \leq r$ for the existential positive formula
\[
    \exists z_1 \ldots \exists z_{r-1}\ \big(E(x,z_1) \wedge
    E(z_1,z_2) \wedge \ldots \wedge E(z_{r-2},z_{r-1}) \wedge E(z_{r-1},y)\big),
\]
stating that two elements in a graph have distance at most $r$. 
Note that we do not require the $z_i$ to be distinct, and hence also elements whose distance is strictly smaller than $r$ can satisfy this formula, in particular in the presence of self-loops.


For a formula $\phi(\bar x)$ and a (colored) graph $G$, we denote by $\phi(G)$ the set all tuples $\bar v$ of vertices of $G$ such that $G$ satisfies $\phi(\bar v)$, 
that is, $\phi(G)=\{\bar v \in V(G)^{|\bar x|}: G\models \phi(\bar v)\}$. 
A \emph{simple interpretation}~$\trans I$ of graphs in (colored) graphs is a pair $(\nu(x), \eta(x,y))$ consisting of two formulas (over the signature of (colored) graphs), where $\eta$ is symmetric, i.e.,
\mbox{$\models \eta(x,y)\leftrightarrow\eta(y,x)$}. 
If $G^+$ is a (colored) graph, then $H=\trans I(G^+)$ is the graph with vertex set $V(H)=\nu(G^+)$ and edge set $E(H)=\eta(G^+)\cap \nu(G^+)^2$.

For a set $\sigma$ of colors, the \emph{coloring operation} $\Gamma_\sigma$ maps a graph $G$ to the set~$\Gamma_\sigma(G)$ of all its $\sigma$-colorings.
A transduction $\trans T$ is the composition
$\trans I\circ\Gamma_\sigma$ of a  coloring operation~$\Gamma_\sigma$ and a simple interpretation $\trans I$ of graphs in $\sigma$-colored graphs.
In other words, for every graph $G$ we have
$\trans T(G)=\{\trans I(H^+): H^+\in\Gamma_\sigma(G)\}$.
In many cases we will deal with transductions where the domain formula~$\nu(x)$ is simply true. In this case we may simply specify the formula $\eta(x,y)$, which by the color predicates it uses fully determines a transduction. We speak of an \emph{$\eta$-transduction}. 
For a class~$\Cc$ of graphs we define $\trans{T}(\CC) := 
\bigcup_{G\in\CC} \trans T(G)$, and we say that a class $\Dd$ is a \emph{transduction} of $\CC$, or that $\CC$ \emph{transduces} $\DD$, if there exists a transduction $\trans T$ such that $\Dd\subseteq \trans{T}(\CC)$. 
A transduction is an \emph{\expos-transduction} if the formulas in the simple interpretation are \expos-formulas. 

\medskip
A class $\Cc$ of graphs is \emph{monadically dependent} if and only if it cannot transduce the class of all graphs, it is \emph{monadically stable} if and only if it cannot transduce the class of all half-graphs~\cite{baldwin1985second}. 
Equivalently, a class is monadically stable if and only if it is monadically dependent and half-graph-free~\cite{nevsetvril2021rankwidth}.  

\medskip
We next collect some properties about transductions. 
It is well-known that transductions compose, see e.g.~\cite{gajarsky2020first}. 
As the next lemma shows the same is true for \expos-transductions. 
For the proof one can follow the lines of the proof in~\cite{gajarsky2020first} and note that the composition of positive existential formulas is again a positive existential formula.

\begin{lemma}\label{lem:expos-transitive}
The class of \expos-transductions is closed under composition. More precisely, let
\[
\trans T_1  =  \trans I_1\circ \Gamma_{\sigma_1}
\qquad\text{and}\qquad
\trans T_2  =  \trans I_2\circ \Gamma_{\sigma_2}
\]
be \expos-transductions with disjoint color sets $\sigma_1,\sigma_2$. 
Then there exists an \expos-transduction $\trans T = \trans I\circ \Gamma_\sigma$, where $\sigma=\sigma_1\cup\sigma_2$, such that for every graph $G$,
\[
\trans T(G) \supseteq (\trans T_2\circ \trans T_1)(G).
\]
\end{lemma}

We further show how to combine transductions that operate on parts of a graph. 


\newcommand{\Glue}{\trans{Glue}}
\begin{lemma}\label{lem:gluing-transductions}
    For every (\expos-)transduction $\trans T$ and $s\in\mathbb{N}$, there is an (\expos-)transduction $\Glue(\trans T,s)$ with the following property. 
    For every graph $G$ and sets $U_1,\ldots, U_s\subseteq V(G)$, we have 
    \[
    \Glue(\trans T,s)(G) \supseteq
    \Big\{\, H_1\cup\cdots\cup H_s \;\Big|\;
    H_i\in \trans T(G[U_i])\text{ for all }i\in [s] \,\Big\}.
    \]
\end{lemma}
\begin{proof}
Assume $\trans T=\mathsf I\circ \Gamma_{\sigma}$ for color set $\sigma$. 
Let $\sigma'$ consist of the following colors:
\begin{itemize}
    \item for each $i\in[s]$, colors $U_i$ and $\overline{U_i}$ (intended to mark (non)membership in the set $U_i\subseteq V(G)$),
    \item for each $i\in[s]$ and each color $C\in\sigma$, a color $C^{(i)}$ (intended to represent the $\sigma$-coloring used to witness an output of $\trans T$ on $G[U_i]$).
\end{itemize}

We define $\Glue(\trans T,s)$ as $\mathsf I'\circ \Gamma_{\sigma'}$ for a suitable interpretation $\mathsf I'=(\nu',\eta')$ of graphs in $\sigma'$-colored graphs.
For each $i\in[s]$, we define new formulas $\nu_i(x)$ and $\eta_i(x,y)$ over the signature of $\sigma'$-colored graphs by recursively replacing in $\nu$ and $\eta$:
\begin{itemize}
    \item every color predicate $C(z)$ by $C^{(i)}(z)$,
    \item every existentially quantified formula $\exists z\ \psi(\bar z, z)$ by $\exists z\ U_i(z) \wedge \psi(\bar z, z)$, and
    \item every universally quantified formula $\forall z\ \psi(\bar z, z)$ by $\forall z\ \overline{U_i}(z) \vee \psi(\bar z, z)$.
\end{itemize}
Intuitively, $\nu_i$ and $\eta_i$ evaluate $\nu$ and $\eta$ inside the induced subgraph $G[U_i]$,
while reading the $\sigma$-colors from the palette $\{C^{(i)}:C\in\sigma\}$.
We define $\mathsf I'=(\nu',\eta')$ by
\[
\nu'(x) \ :=\  \bigvee_{i=1}^s U_i(x) \wedge \nu_i(x),
\qquad\qquad
\eta'(x,y) \ :=\  \bigvee_{i=1}^s U_i(x) \wedge U_i(y) \wedge \eta_i(x,y).
\]
Then $\eta'$ is symmetric because each $\eta_i$ is symmetric (as $\eta$ is) and we only take a disjunction.

\smallskip
Let $G$ be a graph and let $U_1,\ldots,U_s\subseteq V(G)$ be given. 
We show that 
$\Glue(\trans T,s)(G) \supseteq \{\, H_1\cup\cdots\cup H_s : H_i\in \trans T(G[U_i])\text{ for all }i\in [s] \,\}$. 
Consider arbitrary graphs $H_i\in \trans T(G[U_i])$ for $i\in[s]$.
By definition of~$\trans T$, for each $i$ there exists a $\sigma$-colored graph $G_i^+\in\Gamma_\sigma(G[U_i])$ such that $H_i=\mathsf I(G_i^+)$.
Consider the $\sigma'$-coloring $G^+\in\Gamma_{\sigma'}(G)$ that for each $i\in[s]$, colors exactly the vertices of $U_i$ with the color $U_i$ and for each $i\in[s]$ and each $C\in\sigma$, and each vertex $v\in U_i$ satisfies $C^{(i)}(v)$ in $G^+$ if and only if $C(v)$ holds in the $\sigma$-coloring~$G_i^+$ (just as expected). 

It is now easy to see that 
\[
\mathsf I'(G^+) \ =\ H_1\cup\cdots\cup H_s.
\]
Indeed, fix $i\in[s]$. By construction of $G^+$, the induced subgraph $G[U_i]$ together with the colors $\{C^{(i)}:C\in\sigma\}$ coincides with the given $\sigma$-colored graph $G_i^+$.
Moreover, inside $G^+$ the replacements defining $\nu_i$ and $\eta_i$ ensure that $\nu_i$ and $\eta_i$ evaluate exactly like $\nu$ and $\eta$ do in $G_i^+$.
Consequently, the graph with vertex set $\nu_i(G^+)$ and edge set $\eta_i(G^+)\cap \nu_i(G^+)^2$ is precisely $\mathsf I(G_i^+)=H_i$.
Since $\nu'=\bigvee_i \nu_i$ and $\eta'=\bigvee_i \eta_i$, the interpretation $\mathsf I'(G^+)$ is exactly the union $H_1\cup\cdots\cup H_s$.
Hence, $H_1\cup\cdots\cup H_s\in \Glue(\trans T,s)(G)$, proving the desired inclusion.

Finally, it is immediate from the fact that the composition of \expos-formulas is again an \expos-formula, that if~$\trans T$ is an \expos-transduction, then $\Glue(\trans T,s)$ is an \expos-transduction.
\end{proof}


Finally, we note the following fact.
\begin{fact}\label{fact:type-bound}
    For every signature $\sigma$ and $q,t \in \N$, there is a number $N \in \N$, such that there are at most~$N$ many pairwise non-equivalent $\sigma$-formulas of quantifier rank $q$ and with $t$ free variables.
\end{fact}

\section{On self-loops}\label{sec:self-loops}

In this section, we justify the non-standard assumption that we consider graphs with self-loops in the context of positive transductions.

\subsection{Self-loops in general transductions}
We first observe that in the case of general (i.e., not necessarily positive) transductions, whether we forbid/allow/enforce self-loops has no influence on the expressive power of transductions at all.
Here, whether a vertex has a self-loop or not is essentially the same as being marked with a unary predicate: we can distinguish a ``regular'' edge $uv$ from a self-loop by checking that $\neg (u = v)$. Therefore, regarding the input graph of a transduction, self-loops can neither be used to
\begin{itemize}
    \item \ldots conceal information, as the transduction can recover the graph without self-loops, nor
    \item \ldots add new information, as the transduction can add a unary predicate in the coloring step that carries the same information.
\end{itemize}
Moreover regarding its output, by using one of the following three rewritings,
a transduction $\trans{T}_\phi$ can always be assumed to produce output graphs that have \ldots
\begin{enumerate}[label=(R\arabic*)]
    \item\label{rewriting:no-loops} \ldots no self-loops at all, by
    rewriting $\phi(x,y)$ to $\neg (x = y) \wedge \phi(x,y)$, or
    \item\label{rewriting:all-loops} \ldots a self-loop on every vertex, by
    rewriting $\phi(x,y)$ to $(x = y) \vee \phi(x,y) $, or even
    \item\label{rewriting:some-loops} \ldots self-loops on an arbitrary subset $S$ of vertices, by
    rewriting $\phi(x,y)$ to 
    \[
        \big(\neg (x = y) \wedge \phi(x,y)\big) \vee  
        \big( (x = y) \wedge S(x) \big), 
    \]
    where $S$ is added as a unary predicate in the coloring step of the transduction.
\end{enumerate}
Hence, self-loops are not relevant in the setting of general transductions.

\subsection{Self-loops in positive transductions}

In the setting of positive transductions, self-loops naturally appear and seem unavoidable.
Consider for example the transduction $\trans T_\phi$ specified by
\[
    \phi(x,y) := E(x,y) \vee \big(\exists z\ E(x,z) \wedge E(z,y)\big)
\]
that creates the square of a graph.
Applied to a star, this creates self-loops on every vertex, as depicted in \Cref{fig:star}.

\begin{figure}[htbp]
    \centering
    \includegraphics[scale = 1]{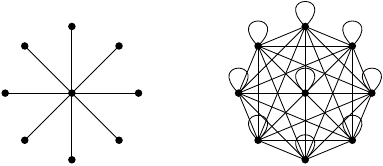}
    \caption{A star and the graph obtained by applying the transduction $\trans{T}_\phi$ to it.}
    \label{fig:star}
\end{figure}

Note that we cannot use \ref{rewriting:no-loops} to make $\phi$ irreflexive, due to the used negation.
Forbidding self-loops would therefore severely restrict our ability to express positive transductions that produce loopless graphs.
For our work the presence of self-loops is crucial.
One could be tempted to modify the definition of a transduction by adding a post-processing step that removes all self-loops from the output graph. Let us call this new notion an \emph{irreflexive-transduction}.
This new notion is no remedy. 
A crucial property of  \expos-transductions is that they are transitive (\Cref{lem:expos-transitive}). For \expos-irreflexive-transductions transitivity fails.
The example in \Cref{fig:loops-co-matching} shows how to create co-matchings from matchings by chaining three \expos-irreflexive-transductions, which is not possible by a single \expos-irreflexive-transduction.
The second transduction is defined by the formula $\phi(x,y) := \exists z\ (E(x,z) \wedge E(z,y))$.
Note that the first \expos-irreflexive-transduction creates a large irreflexive clique from an independent set, which is not possible with an \expos-transduction.

\begin{figure}[htbp]
    \centering
    \includegraphics[scale = 1]{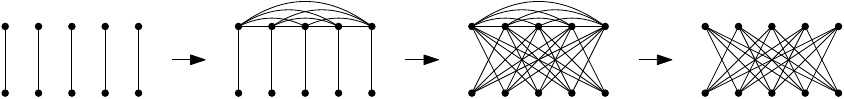}
    \caption{A chain of three \expos-irreflexive-transductions that produce the class of co-matchings from the class of matchings.}
    \label{fig:loops-co-matching}
\end{figure}

As mentioned in the introduction, \Cref{thm:nep-simple} fails in the absence of self-loops:
in the class of cliques without self-loops (this class is semi-ladder-free and monadically dependent) the simple formula $\phi(x,y) = E(x,y)$ has the monadic non-equality-property.


The ``philosophical reason'' as to why the absence self-loops involuntarily increases the expressive power of \expos-formulas is the following:
An edge between two vertices that do not have self-loops carries the additional information that these two vertices are non-equal; an inherently non-positive property.

\section{Subflips}\label{sec:subflips}

In this section, we introduce the new notion of a \emph{subflip}.
For this purpose, let us first review the notion of a \emph{flip} from the literature.

\subsection{Flips and partitions}

Let $\PP$ be a partition of a set $V$.
For every element $v \in V$, we denote by $\PP(v) \in \PP$ the part of $\PP$ that contains $v$.
For a set $S \subseteq V$ we denote by 
\[
    \PP|_S := \{ P \cap S : P \in \PP \}
\]
the \emph{restriction} of $\PP$ to $S$, which is a partition of $S$.
For another partition $\QQ$ on $V$ we denote by
\[
    \PP \wedge \QQ := \{P \cap Q : P \in \PP, Q\in \QQ\}
\]
the \emph{common refinement} of $\PP$ and $\QQ$.

\medskip
Fix a graph~$G$ and a partition~$\PP$ of its vertices.
Let~$F \subseteq \PP^2$ be a symmetric relation.
We define~$H := G \oplus (\PP,F)$ to be the graph with vertex set~$V(G)$,
where we define the edges between distinct vertices \mbox{$u,v\in V(G)$} using the following condition:
\[
    uv \in E(H) \Leftrightarrow \begin{cases}
        uv \notin E(G) & \text{if } (\PP(u), \PP(v)) \in F,\\
        uv \in E(G) & \text{otherwise.}
    \end{cases}
\]
Self-loops are preserved, i.e., for every $v\in V(G)$ we have $vv \in E(H) \Leftrightarrow vv \in E(G)$.
We call~$H$ a \emph{$\PP$-flip} of~$G$.
If~$\PP$ has at most~$k$ parts, we say that~$H$ a \emph{$k$-flip} of~$G$. See \Cref{fig:flips} in the introduction for an example.
We note a few basic facts about flips. 

\begin{fact}\label{lem:flip-basics}
    Let $G$, $H$, $H_1$, $H_2$ be graphs on a vertex set $V$, let
    $\PP,\PP_1, \PP_2$ be partitions of $V$, and let~$S \subseteq V$.
    \begin{enumerate}
        \item Symmetry: If $H$ is a $\PP$-flip of $G$, then $G$ is also a $\PP$-flip of $H$.
        \item Transitivity: If $H_1$ is a $\PP_1$-flip of $G$ and $H_2$ is a $\PP_2$-flip of $H_1$, then $H_2$ is a $(\PP_1 \wedge \PP_2)$-flip of $G$.
        \item Hereditariness: If $H$ is a $\PP$-flip of $G$, then $H[S]$ is a $\PP|_S$-flip of $G[S]$.
    \end{enumerate}
\end{fact}

\subsection{Subflips}

\begin{definition}[Fully (non-)adjacent sets]
For a graph $G$ and two sets $A, B \subseteq V(G)$, with possibly $A = B$, we say $A$ and $B$ are \emph{fully adjacent} (\emph{fully non-adjacent}) if every two distinct vertices $a\in A$ and $b \in B$ are adjacent in $G$ (non-adjacent in $G$).
If $A$ and $B$ are disjoint, this means they semi-induce a biclique in $G$ (an edgeless graph in $G$).
If $A = B$, this means $A$ induces a clique (an independent set).    
\end{definition}

Because we only speak about distinct vertices, the above definition makes no assumptions about self-loops.

\begin{definition}[Subflips]
    For a graph $G$ and partition $\PP$ of $V(G)$, we define the \emph{$\PP$-subflip of $G$} to be the graph
    \begin{center}
        $G \ominus \PP := G \oplus (\PP,F)$, where $F := \{(A,B) : \text{$A$ and $B$ are fully adjacent in $G$}\}$. 
    \end{center}
    If~$\PP$ has at most~$k$ parts, we say that~$H$ a \emph{$k$-subflip} of~$G$.
\end{definition}

Note that the structure of subflip differs from that of a flip. We fix a partition $\PP$ but do not supply a flip relation $F$. 
Instead, we always use the relation $F$ that removes the most edges.
The following crucial observation gives the subflips their name.
\begin{observation}
    If $H$ is a subflip of $G$, then $H$ is a subgraph of $G$.
\end{observation}

We next establish a basic toolbox for working with subflips.
This will be guided by the symmetry, transitivity, and hereditariness properties that hold for flips (\Cref{lem:flip-basics}).

\subsection{Symmetry}
Clearly, the subflip relation is not symmetric: the independent set $nK_1$ is a subflip of the clique~$K_n$, but the reverse fails because $K_n$ is no subgraph of $nK_1$.
To reverse the subflip operation, we introduce \emph{pure flips} as a middle ground between subflips and regular flips.

\begin{definition}[Pure flips]
    We call a flip $G \oplus (\PP,F)$ of $G$ \emph{pure}, if every pair of parts $(A,B) \in F$ is fully adjacent or fully non-adjacent in $G$.
\end{definition}

\begin{observation}
    If $H$ is a $k$-subflip of $G$, then $H$ and $G$ are pure $k$-flips of each other.
\end{observation}

Crucially, in reflexive graphs we can transduce pure flips using quantifier-free, positive formulas.

\begin{lemma}[\needsselfloops]\label{lem:trans-pure-flips-reflexive}
    For every $k\in \N$ there is a quantifier-free, positive formula $\phi(x,y)$ such that for every reflexive graph $G$ and pure $k$-flip $H$ of $G$, $H$ $\phi$-transduces $G$.
\end{lemma}

Instead of proving \Cref{lem:trans-pure-flips-reflexive}, we prove the following more general lemma, where we only require the vertices in parts that are flipped with themselves to have self-loops.

\begin{lemma}\label{lem:trans-pure-flips}
    For every $k\in \N$ there is a quantifier-free, positive formula $\phi(x,y)$ with the following property.
    Let $G$ be a partially reflexive graph $G$ and $H = G \oplus (\PP,F)$ be a pure $k$-flip of $G$ satisfying
    \begin{equation}\label{eq:tame-loops}
        \forall v \in V(G): \quad (\PP(v),\PP(v)) \in F \Rightarrow vv \in E(G).\tag{$*$}
    \end{equation}
    Then $H$ $\phi$-transduces $G$.
\end{lemma}

\begin{proof}
    We have to show that there exists a formula $\phi$ and a coloring $H^+$ of $H$ such that for all $u,v \in V(G)$ we have $H^+ \models \phi(u,v) \Leftrightarrow uv \in E(G)$.
    Let $H^+$ be the coloring of $H$ obtained by marking the parts of~$\PP$.
    We partition $F$ into $F_\top$ and $F_\bot$ where $F_\top$ contains the pairs in $F$ that are fully adjacent in $G$ and $F_\bot$ contains those that are fully non-adjacent in $G$.
    Consider the positive, quantifier-free formula
    \begin{align*}
        \psi(x,y) &:= \alpha(x,y) \vee \beta(x,y)  \vee \gamma(x,y),\\
        \alpha(x,y) &:= (x = y) \wedge E(x,x), \\
        \beta(x,y) &:= \bigvee_{(P,Q) \in F_\top} (x \in P \wedge y \in Q),\\
        \gamma(x,y) &:= \bigvee_{(P,Q) \notin F_\bot} (E(x,y) \wedge x \in P \wedge y \in Q).
    \end{align*}
    \begin{claim}
        For all $u,v \in V(G)$: $H^+ \models \psi(u,v) \Leftrightarrow uv \in E(G)$.
    \end{claim}
    \begin{claimproof}
    Let $P := \PP(u)$ and $Q := \PP(v)$.
    \begin{enumerate}
        \item Assume $uv \in E(G)$.
        The goal is to show that $H^+$ satisfies one of $\alpha(u,v)$, $\beta(u,v)$, and $\gamma(u,v)$.
        \begin{enumerate}
            \item Assume $u = v$. As subflips preserve self-loops, $u$ and $v$ are also adjacent in $H$ and $\alpha$ is satisfied.
            \item Assume $u \neq v$.
            
            \begin{enumerate}
                \item Assume $(P,Q) \in F$. 
                Since $u$ and $v$ are adjacent in $G$, the parts $P$ and $Q$ must be fully adjacent in $G$ by definition of a pure flip. Then $(P,Q) \in F_\top$ and $\beta$ is satisfied.
                \item 
        Assume $(P,Q) \notin F$. Then $u$ and $v$ are adjacent also in $H$.
        Moreover, $(P,Q) \notin F_\bot$, so $\gamma$ is satisfied.
            \end{enumerate}

        \end{enumerate}

        \item Assume $uv \notin E(G)$. The goal is to show that $H^+$ satisfies none of $\alpha(u,v)$, $\beta(u,v)$, and $\gamma(u,v)$. 
        
        \begin{enumerate}
            \item Assume $u = v$. 
            As subflips preserve self-loops, $u$ and $v$ are also non-adjacent in $H$ and $\alpha$ and $\gamma$ are both not satisfied.
            By \eqref{eq:tame-loops}, we have that $(P, Q=P) \notin F \supseteq F_\top$ and $\beta$ is not satisfied, either.

            \item Assume $u \neq v$. Clearly, $\alpha$ is not satisfied.
            
            \begin{enumerate}
                \item Assume $(P,Q) \in F$. Since $u$ and $v$ are non-adjacent in $G$, the parts $P$ and $Q$ must be fully non-adjacent in $G$ by definition of a pure flip. Then $(P,Q) \notin F_\top$ and $(P,Q) \in F_\bot$, and neither $\beta$ nor $\gamma$ are satisfied.
                \item Assume $(P,Q) \notin F \supseteq F_\top$.
                Then $\beta$ is not satisfied. Moreover, $u$ and $v$ are also non-adjacent in $H$ so $\gamma$ is not satisfied.
            \end{enumerate}
        \end{enumerate}
    \end{enumerate}
    Having exhaustively checked all cases, this proves the claim.
    \end{claimproof}

    The presented formula $\psi$ depends not only on $k$, but also on the choice of $F$ and its partition into $F_\top$ and $F_\bot$.
    Note that there are at most $2^{k^2}$ ways to choose and $F_\top$ and $F_\bot$ each. By increasing the color palette, we can encode this information into the colors of every vertex in $H^+$ and construct the desired formula $\phi$ which branches over the bounded number of cases and only depends on $k$.
\end{proof}

To relax the requirements on self-loops a little bit, we show that we can always add a constant number of vertices with arbitrary neighborhoods (one by one).

\begin{lemma}\label{lem:add-one}
    For every quantifier-free, positive formula~$\phi$ there is a quantifier-free, positive formula~$\psi$ such that for every two partially reflexive graphs 
    $G$ and $H$ on the same vertex set $A \uplus \{b\}$: if~$H[A]$ $\phi$-transduces~$G[A]$, 
    then $H$ $\psi$-transduces $G$.
\end{lemma}
\begin{proof}
    We extend the coloring of $H[A]$ that is used in the $\phi$-transduction to a coloring of $H$ by additionally:
    \begin{itemize}
        \item marking the set $A$ with a color predicate $A$,
        \item marking the vertex $b$ with a color predicates $C_b$,
        \item marking the neighborhood of $b$ in $G$ with a color predicate $N_b$.
    \end{itemize}
    This coloring of $H$ witnesses that $H$ $\psi$-transduces $G$ for the formula
    \[
        \psi(x,y) := (x \in C_b \wedge y \in N_b) \vee (y \in C_b \wedge x \in N_b) 
        \vee (x \in A \wedge y \in A \wedge \phi(x,y)).\qedhere
    \]
\end{proof}

The lemmas proven so far culminate in the following corollary, which we later use to characterize the monadic non-equality property.

\begin{corollary}\label{cor:transduce-subflips-partially-reflexive}
    For all $k,t\in \N$ there is a quantifier-free, positive formula $\phi(x,y)$ with the following property.
    Let $G$ be a partially reflexive graph $G$ that excludes the irreflexive $K_t$ as an induced subgraph, and let $H$ be a $k$-subflip of $G$.
    Then $H$ $\phi$-transduces $G$.
\end{corollary}
\begin{proof}
    Let $\PP$ be a partition of size at most $k$ such that $H = G \ominus \PP$ and let $F \subseteq \PP \times \PP$ be the pairs of parts that are fully adjacent. Then $H = G \oplus (\PP,F)$ and also $G = H \oplus (\PP,F)$.
    Let $\PP^* := \{P\in \PP : (P,P)\in F\}$ be the set of parts flipped with itself.
    By definition, each part of $\PP^*$ forms a clique in $G$.
    Since the irreflexive~$K_t$ is excluded, this means the size of the set $A$ of vertices without self-loops that are contained in a part of $\PP^*$ is bounded by $kt$. Let $B := V(G) - A$ and notice that $H[B]$ and $G[B]$ satisfy the conditions of \Cref{lem:trans-pure-flips}: $G[B]$ is a pure $k$-flip of $H[B]$ and in the witnessing partition $\PP|_B$ no part that is flipped with itself contains a vertex without a self-loop anymore.
    This means there is a formula $\psi$ (depending only on $k$) such that~$G[B]$ $\psi$-transduces $H[B]$.
    Iterating \Cref{lem:add-one}, we can add back the at most $kt$ removed vertices~$A$ to get a quantifier-free, positive formula $\phi$ such that $H$ $\phi$-transduces $G$.
\end{proof}

\subsection{Transitivity}

The transitivity property of general flips does not hold for subflips. The finer parts of $\PP_1 \wedge \PP_2$ can be fully-adjacent, while their containing coarser parts in $\PP_1$ and $\PP_2$ are not. In this case the subflip specified by $\PP_1 \wedge \PP_2$ removes even more edges.
We can still work with refinements, using the following property.

\begin{lemma}\label{lem:refine}
    For every graph $G$, partition $\PP$ of $V(G)$, and refinement $\RR$ of $\PP$ we have
    \[
        (G\ominus \PP) \ominus \RR = G \ominus \RR.
    \]
\end{lemma}
\begin{proof}
    Let $v_1$ and $v_2$ be distinct vertices in $G$ and let $P_i \in \PP$ and $R_i \in \RR$ be the parts containing $v_i \in R_i \subseteq P_i$.
    If $v_1v_2$ is a non-edge in $G$ then this also true in the two subgraphs $(G \ominus \PP) \ominus \RR$ and $G \ominus \RR$ and there is nothing to check.
    Assume therefore $v_1v_2$ an edge in $G$.

    Assume $v_1v_2$ is removed in $G \ominus \RR$.
    Then the parts $R_1$ and $R_2$ are fully adjacent in $G$.
    If $P_1$ and $P_2$ are fully adjacent in $G$, then $v_1v_2$ is removed already in $G \ominus \PP$.
    Otherwise, the adjacency between $P_1$ and $P_2$ is the same in $G$ as in $G \ominus \PP$.
    In particular the subparts $R_1$ and $R_2$ remain fully adjacent in $G \ominus \PP$, which means that $v_1v_2$ is removed in  $(G \ominus \PP) \ominus \RR$.

    Assume $v_1v_2$ is removed in $(G \ominus \PP) \ominus \RR$.
    If $v_1v_2$ is already removed in $G \ominus \PP$, we are done as $G\ominus \RR$ is a subgraph of $G \ominus \PP$.
    Otherwise, $v_1v_2$ is still present in $G \ominus \PP$ and the adjacency between~$R_1$ and $R_2$ is the same in $G\ominus \PP$ and $G$.
    Then applying $\ominus \RR$ must have removed the edge $v_1v_2$ in $G \ominus \PP$ and also in $G$.
\end{proof}

\subsection{Hereditariness}

Taking a subflip does not commute with taking an induced subgraph, as shown by the following example.

\begin{example}
    For the graph $G := G_1 \uplus G_2 = 2K_2$, where $G_1$ and $G_2$ are both the single edge graph $K_2$, the partition $\PP: = \{V(G)\}$ into a single part, and the set $S := V(G_1)$ we have
    \[
        2K_1 = G[S] \ominus \PP|_S \neq (G \ominus \PP) [S] = K_2.
    \]
\end{example}

We will work around this using the following lemma.

\begin{lemma}\label{lem:subflip-induce}
    For every graph $G$, partition $\PP$ of $V(G)$, and $S\subseteq V(G)$
    \[
        G[S] \ominus \PP|_S = ((G \ominus \PP) [S]) \ominus \PP|_S.
    \]
\end{lemma}

\begin{proof}
    We can restrict our attention to the set $S$.
    Non-edges in $G[S]$ stay non-edges in the two subflips.
    Hence, it suffices to prove that edges of $G[S]$ are removed in 
    $G[S] \ominus \PP|_S$ if and only if they are removed in $((G \ominus \PP) [S]) \ominus \PP|_S$.
    Fix any edge $uv$ in $G[S]$.

    Assume first that $uv$ is a non-edge $G[S]\ominus \PP|_S$.
    Then $\PP|_S(u)$ and $\PP|_S(v)$ are fully adjacent in~$G[S]$.
    If the coarser parts $\PP(u)$ and $\PP(v)$ are also fully adjacent in $G$, then $uv$ was already removed in~$G \ominus \PP$.
    Otherwise, the adjacencies between $\PP|_S(u)$ and $\PP|_S(v)$ are the same in $G[S]$ and $(G \ominus \PP) [S]$, so~$uv$ is removed in 
    $((G \ominus \PP) [S]) \ominus \PP|_S$.

    \pagebreak
    Assume now that $uv$ is a non-edge in $((G \ominus \PP) [S]) \ominus \PP|_S$.
    Then either 
    \begin{enumerate}
        \item $\PP(u)$ and $\PP(v)$ are fully adjacent in $G$, or
        \item $\PP|_S(u)$ and $\PP|_S(v)$ are fully adjacent in $(G\ominus \PP)[S]$.
    \end{enumerate}
    In the first case, in particular $\PP|_S(u)$ and $\PP|_S(v)$ are fully adjacent in $G[S]$, and we are done.
    If the first case fails and the second holds, then the adjacencies between $\PP|_S(u)$ and $\PP|_S(v)$ are the same in $G[S]$ and $(G \ominus \PP) [S]$, so $uv$ is removed in 
    $G[S] \ominus \PP|_S$.
\end{proof}
\subsection{Approximating flips through subflips}

A \emph{semi-induced matching} of size $t$ in a graph $G$ is a set of vertices $a_1,\ldots,a_t$ and $b_1,\ldots,b_t$ such that for all $i,j \in [t]$: $E(a_i,b_j) \in E(G) \Leftrightarrow i = j$.
In the more restricted notion of an \emph{induced matching}, we additionally demand that all the $a_i$s form an independent set, and all the $b_i$s form an independent set.
Distinguishing between the two will be important when making the following observation.

\begin{observation}\label{obs:induced-vs-semi-induced}
    Let $G$ be a graph, $B$ be a bipartite graph semi-induced in $G$, and $\overline{B}$ be the bipartite complement of $B$.
    If $G$ has co-matching-index less than $t$, then $\overline{B}$ contains no induced matching of size $t$.
\end{observation}

The observation fails if we replace ``induced matching'' with ``semi-induced matching'' as witnessed by the following graph $G$ on the left.

\begin{figure}[htbp]
    \centering
    \includegraphics{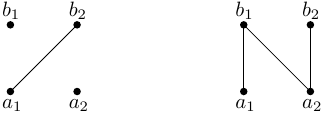}
\end{figure}

Since $G$ has only one edge, it has co-matching index less than $2$. Depicted on the right is $\overline{B}$: the bipartite complement of the graph $B$ semi-induced between $\{a_1,a_2\}$ and $\{b_1,b_2\}$.
$\overline{B}$ contains a semi-induced matching of size $2$, obtained by semi-inducing between $\{a_1,b_2\}$ and $\{b_1,a_2\}$. Hence, $G$ is a counterexample to \Cref{obs:induced-vs-semi-induced} for semi-induced matchings.
This is the reason, we state the next lemma for forbidden induced matchings, even though for forbidden semi-induced matchings a better bound of $2t$ would be possible.

\begin{lemma}\label{lem:matching-diam}
    Let $G$ be a graph that contains no induced matching of size $t$. 
    Then $G$ contains less than~$t$ connected components of size at least $2$ and each of these components has diameter less than~$3t$.
\end{lemma}
\begin{proof}
    Let $C_1, \ldots, C_m$ be the connected components of $G$ of size at least $2$.
    Each component $C_i$ contains at least one edge $u_i v_i$.
    The $u_i$s and $v_i$s induce a matching of order $m$.
    Hence, $m < t$.

    Assume towards a contradiction that one of these components has diameter at least $3t$. This is witnessed by an induced path with $3t$ edges of which $t$ induce a matching of order $t$, a contradiction.
\end{proof}

\begin{lemma}\label{lem:transfer}
    For every graph $G$ with co-matching-index less than $t$ and size-$k$ partition $\PP$ of~$V(G)$,
    there exists a refinement $\QQ$ of $\PP$ of size at most $k \cdot t^k$, such that for every $\PP$-flip $H$ of $G$ and all vertices $u,v\in V(G)$
    \[
        G \ominus \QQ \models E(u,v)
        \text{ implies }
        \dist_H(u,v) < 3t.
    \]
\end{lemma}

\begin{remark}
    \Cref{lem:transfer} implies for every $r\in\NN$ and $v\in V(G)$:
    \[
        B_r^{G \ominus \QQ}(v) \subseteq B_{3tr}^H(v).
    \]

\end{remark}

\begin{proof}[Proof of \Cref{lem:transfer}]
    We construct $\QQ$ by individually refining each part $P \in \PP$.
    For all pairs of parts $P,Q \in \PP$ (possibly $P = Q$) we will define a partition $\XX(P,Q)$ of $P$ of size at most $t$,
    which naturally lifts to a partition $\XX'(P,Q) := \XX(P,Q) \cup \{V(G) \setminus P\}$ of the whole vertex set $V(G)$.
    We define $\QQ$ to be the common refinement of all the $\XX'(P,Q)$.
    As each $\XX'(P,Q)$ only refines a single part $P$, this bounds the size of $\QQ$ by $k \cdot t^k$.

    \medskip\noindent
    \emph{Construction of $\XX(P,Q)$ for $P \neq Q$.}
        Consider the bipartite graph $B$ semi-induced by the parts~$P$ and $Q$ in~$G$ and let $\overline{B}$ be its bipartite complement.
        First, we add to $\XX(P,Q)$ the \emph{isolation part} $X_0(P,Q)$ containing all vertices of $P$ that are isolated in $\overline{B}$.
        Second, let $C_1, \ldots, C_m$ be the connected components of $\overline{B}$ of size at least two.
        For each $C_i$, we add to $\XX(P,Q)$ the \emph{component part} $X_i(P,Q) := C_i \cap P$.
        Since $G$ has co-matching-index less than $t$, $\overline{B}$ cannot contain an induced matching of size~$t$ (\Cref{obs:induced-vs-semi-induced}).
        By \Cref{lem:matching-diam}, we conclude that $m < t$ and $|\XX(P,Q)| = m + 1 \leq t$.

    \medskip\noindent
    \emph{Construction of $\XX(P,P) =: \XX(P)$.} 
    Consider the complement $\overline{G[P]}$ of the graph $G[P]$ induced by $P$ in~$G$.
    We proceed as in the case above: we add to $\XX(P)$ the \emph{isolation part} $X_0(P)$ containing all vertices isolated in $\overline{G[P]}$ and each connected component $C_i$ of $\overline{G[P]}$ of size at least two gets its own \emph{component part} $X_i(P) := C_i$.
    Again by \Cref{lem:matching-diam}, $|\XX(P)|  \leq t$. 

    \medskip\noindent
    \emph{Analysis for edges between distinct parts.}
    Consider an edge $uv$ in $G \ominus \QQ$ with endpoints in distinct parts $u \in P \in \PP$ and $v \in Q \in \PP$ and an arbitrary $\PP$-flip $H$ of $G$. 
    Our goal is to show that $\dist_H(u,v) < 3t$.
    If $u$ and $v$ are also adjacent in $H$ we are done, so assume they are non-adjacent.
    As they are adjacent in the subgraph $G \ominus \QQ$ of $G$, they must be also adjacent in $G$.
    This means the adjacency between~$P$ and $Q$ was flipped in $H$ and the bipartite graph $\overline{B}$ semi-induced between $P$ and $Q$ in $H$ is the same as the bipartite complement of the graph $B$ semi-induced between $P$ and $Q$ in $G$.
    Note that the isolation part $X_0(P,Q)$ and $Q$ are fully non-adjacent in $\overline{B}$, which means they are fully adjacent in $G$, which means they are fully non-adjacent in $G \ominus \QQ$.
    As $u$ is adjacent to $v \in Q$ in $G \ominus \QQ$, we must have $u \notin X_0(P,Q)$ and symmetrically $v \notin X_0(Q,P)$.
    Then $u$ and $v$ must be contained in component parts $u \in X_i(P,Q)$ and $v \in X_j(Q,P)$.
    If $u$ and $v$ are contained in different connected components of $\overline{B}$, then $X_i(P,Q)$ and $X_j(Q,P)$ are fully non-adjacent in $\overline{B}$ and also in $G \ominus \QQ$; a contradiction.
    This means $u$ and $v$ are contained in the same connected component of $\overline{B}$.
    By \Cref{lem:matching-diam}, this component has diameter less than $3t$ in $\overline B$, which yields the desired bound  $\dist_H(u,v) < 3t$.

    \medskip\noindent
    \emph{Analysis for edges inside the same part.}
    Similar as above. If $\{u,v\} \subseteq P \in \PP$ are adjacent in $G\ominus \QQ$ but non-adjacent in $H$,
    then the part $P$ was flipped with itself.
    Every two distinct parts of $\XX(P)$ are fully non-adjacent in $\overline{G[P]}$, so fully adjacent in $G$, so fully non-adjacent in $G \ominus \QQ$.
    This means~$u$ and~$v$ must be contained in the same part of $\XX(P)$. 
    The isolation part $X_0(P)$ forms a clique in~$G$ and therefore an independent set in $G \ominus \QQ$.
    Hence, $u$ and $v$ must be contained in the same component part, and we conclude by \Cref{lem:matching-diam}.
\end{proof}

\begin{lemma}\label{lem:preserve-cmi}
    For every graph $G$ and partition $\PP$ of $V(G)$ the following holds:
    \begin{enumerate}
        \item For all $t>1$, if $\operatorname{co-matching-index}(G \ominus \PP)\geq t \cdot |\PP|^2$, then $\operatorname{co-matching-index}(G) \geq t$.
        
        \item For all $t\geq 0$, if $\operatorname{co-matching-index}(G) \geq t \cdot |\PP|^2$, then 
        $\operatorname{co-matching-index}(G \ominus \PP) \geq t$.
    \end{enumerate}
\end{lemma}
\begin{proof}
    For the first implication let $A,B \subseteq V(G)$ be sets of size $t \cdot |\PP|^2$ that semi-induce a co-matching in $G \ominus \PP$.
    Applying the pigeonhole principle twice, we get two parts $P,Q \in \PP$ and subsets $A' \subseteq A \cap P$ and $B' \subseteq B \cap Q$ with $|A'| = |B'| = t$ that still semi-induce a co-matching in $G \ominus \PP$. 
    Hence, $P$ and~$Q$ are neither fully-adjacent in $G \ominus \PP$ nor in $G$ (here we used $t > 1$, as a co-matching of size~$1$ would just be a non-edge, so $P$ and $Q$ could be fully adjacent).
    This means the adjacency between~$A'$ and~$B'$ is the same in~$G$ and $G \ominus \PP$, where they semi-induce a co-matching of size $t$.

    For the second implication with $t > 1$, swap the roles of $G\ominus \PP$ and $G$ in the proof above.
    The case $t= 0$ is vacuous and for $t=1$ notice that every co-matching of size $1$ in $G$ (i.e., every non-edge in $G$) is also a co-matching of size $1$ in $G\ominus \PP$.
\end{proof}

\section{Subflip-flatness}\label{sec:sub-flatness}

 \begin{definition}[(Sub)flip-flatness]
    For $r\in \NN \cup \{ \infty \}$, a graph class $\Cc$ is \emph{$r$-flip-flat} (\emph{$r$-subflip-flat}), if there exists a $\emph{margin}$ function $M:\NN \to \NN$ and a \emph{budget} $k \in \NN$ such that for every $G \in \Cc$ and set $W \subseteq V(G)$ of size at least $M(m)$, there is a $k$-flip ($k$-subflip) $H$ of $G$ and a size $m$ set $A \subseteq W$ such that for all distinct $u,v \in A$ we have
    \[
        \dist_H(u,v) > r.
    \]
    (For $r=\infty$, we demand that $u$ and $v$ are in different connected components of $H$.)
 \end{definition}

Dreier, Mählmann, Siebertz, and Toru\'nczyk in~\cite{dreier2022indiscernibles} proved the following combinatorial characterization of monadic stability.

\begin{theorem}[\cite{dreier2022indiscernibles}]\label{thm:flatness}
    For every graph class $\CC$, the following are equivalent:
    \begin{enumerate}
        \item $\CC$ is monadically stable.
        \item $\CC$ is $r$-flip-flat for every $r\in \N$.
    \end{enumerate}
\end{theorem}

The goal of this section is to prove the following analog for co-matching-free classes.

\begin{theorem}\label{thm:subflatness}
    For every graph class $\CC$, the following are equivalent:
    \begin{enumerate}
        \item $\CC$ is monadically stable and co-matching-free.
        \item $\CC$ is $r$-subflip-flat for every $r\in \N$.
    \end{enumerate}
\end{theorem}

The forward direction is a simple application of \Cref{lem:transfer}.

\begin{lemma}\label{lem:subflatness-forward}
    For every graph class $\Cc$ with co-matching-index less than $t \in \N$ and for every $r\in \N \cup \{\infty\}$:
    If~$\Cc$ is $3tr$-flip-flat, then it is also $r$-subflip-flat.
\end{lemma}
\begin{proof}
Let $M : \NN \rightarrow \NN$ and $k\in \NN$ be the margin and budget for which $\Cc$ is $3tr$-flip-flat.
We claim that~$\Cc$ is $r$-subflip-flat for margin $M$ and budget $k' := k \cdot t^k$.
Given a size $M(m)$ vertex subset~$W$ in a graph $G\in\Cc$, we can apply $3tr$-flip-flatness, which yields a size $m$ subset $A$ of $W$ and a $k$-flip~$H$ in which all elements from $A$ are at pairwise distance greater than $3tr$ from each other.
Using \Cref{lem:transfer}, we obtain from $H$ a $k t^k$-subflip $H'$ of $G$ in which all elements from $A$ are at pairwise distance greater than $r$ from each other.
(In particular if $r = \infty$, then elements in different connected components of $H$ are still in different connected components of $H'$.)
\end{proof}

For the backward direction, notice that if a graph class is not monadically stable, then, since it is not even flip-flat, it is for sure not subflip-flat.
It remains to prove that classes with unbounded co-matching-index are not subflip-flat.

\begin{lemma}\label{lem:subflatness-backward}
    No class with unbounded co-matching-index is $2$-subflip-flat.
\end{lemma}
\begin{proof}
    Assume towards a contradiction the existence of a graph class $\CC$ with unbounded co-matching-index that is $2$-subflip-flat with margin $M$ and budget $k$.
    Let $G \in \CC$ be a graph with co-matching-index at least $\ell := M(3k^2)$ and let $a_1,\ldots,a_\ell$ and $b_1,\ldots,b_\ell$ be the vertices that semi-induce a co-matching of order $\ell$.
    Applying subflip-flatness to $G$ and the set $W = \{a_1,\ldots,a_\ell\}$ yields a partition~$\PP$ of size at most $k$ and a size $\ell' := 3k^2$ subset $A$ of $W$ (we can assume $A = \{ a_1, \ldots, a_{\ell'} \}$) such that the vertices in $A$ are at pairwise distance greater than $2$ in $G\ominus \PP$.
    Applying the pigeonhole principle twice, we can assume $a_1,a_2,a_3$ are in the same part $P$ of $\PP$ and $b_1,b_2,b_3$ are in the same part $Q$ of $\PP$.
    This means $P$ and $Q$ are not fully adjacent in $G$, so the edges between $P$ and $Q$ were not deleted in the subflip.
    This means the six vertices still semi-induce a co-matching in $G \ominus \PP$.
    Then $a_1$ and $a_2$ have distance at most two from each other; a contradiction.
\end{proof}

This finishes the proof of \Cref{thm:subflatness}.
We later want to use subflip-flatness not only for singleton elements, but also for tuples.
We call two tuples \emph{disjoint} if they share no element.
We define the distance between two $t$-tuples $\bar u$ and $\bar v$ in $H$ as $\min_{i,j\in [t]} \dist_H(u_i,v_j)$.

\begin{definition}[Subflip-flatness for tuples]
For $r\in \NN \cup \{ \infty \}$, a graph class $\Cc$ is $r$-subflip-flat \emph{for tuples}, if for every $t\in\N$ there exists a margin function $M:\NN \to \NN$ and a budget $k \in \NN$ such that for every $G \in \Cc$ and set $W \subseteq V(G)^t$ of disjoint $t$-tuples of size at least $M(m)$, there is a $k$-subflip $H$ of $G$ and a size $m$ set $A \subseteq W$ such that for all distinct $\bar u, \bar v \in A$ we have
\[
    \dist_H(\bar u, \bar v) > r.
\]
\end{definition}

\begin{lemma}\label{lem:sff-tuples}
    Every graph class that is $2r$-subflip-flat is also $r$-subflip-flat for tuples.
\end{lemma}
\begin{proof}
    The lemma is proved exactly as Lemma 2.10 of~\cite{pilipczuk2018number}, replacing uniform quasi-wideness by subflip-flatness and using \cref{lem:refine}.
\end{proof}
\section{Subflipper-rank}\label{sec:subflipper-rank}

In this section we will characterize the semi-ladder-free fragment of monadic dependence as those classes that have bounded \emph{subflipper-rank}.

\subsection{Definition and basic properties}
\newcommand{\frk}{\mathrm{frk}}
\newcommand{\sfrk}{\mathrm{sfrk}}
\newcommand{\Ball}{\mathrm{Ball}}

Given a graph $G$, an integer $r$, and a vertex $v \in V(G)$ we write $\Ball^G_{r}(v)$ as a shorthand for the graph $G[N_r^G[v]]$ induced by the (closed) $r$-neighborhood of $v$ in $G$.

\begin{definition}[(Sub)flipper-rank]
    Fix $r\in \NN$ and $k\in\NN$. For the single vertex graph $K_1$ 
    we define the \emph{flipper-rank}
    \[
        \frk_{r,k}(K_1) := 0.
    \]
    For every other graph $G$ we define
    \[
        \frk_{r,k}(G) := 1+ \min_{\substack{\text{$H$ is a} \\ \text{$k$-flip of $G$}}} \max_{v \in V(H)} \frk_{r,k}(\Ball_r^H(v)).
    \]
    A graph class $\CC$ has \emph{bounded $r$-flipper-rank} if there are $k,\ell \in \NN$ such that $\frk_{r,k}(G) \leq \ell$ for all $G \in \CC$.
    The definition of the \emph{subflipper-rank} $\sfrk_{r,k}(G)$ is derived by demanding that $H$ is also a $k$-subflip of $G$. 
    We define \emph{bounded $r$-subflipper-rank} as expected.
\end{definition}

Intuitively, we can understand the flipper-rank (and also the subflipper-rank) as a game called the \emph{flipper-game}~\cite{flipper-game}.
Played on a graph $G$ between two players \emph{flipper} and \emph{localizer}, in every round of the game, flipper applies a $k$-flip to the current graph (with the goal of decomposing it as fast as possible), and localizer shrinks the arena to an $r$-neighborhood (with the goal of surviving as long as possible).
The flipper-rank then denotes the number of rounds needed for flipper to finish the game.
The following theorem is a key ingredient for the first-order model checking algorithm for monadically stable classes.

\begin{theorem}[\cite{flipper-game}]\label{thm:flipper}
    For every graph class $\CC$, the following are equivalent:
    \begin{enumerate}
        \item $\CC$ is monadically stable.
        \item $\CC$ has bounded $r$-flipper-rank for every $r\in \N$.
    \end{enumerate}
\end{theorem}

The goal of this section is to prove the following analog for co-matching-free classes.

\begin{theorem}\label{thm:subflipper}
    For every graph class $\CC$, the following are equivalent:
    \begin{enumerate}
        \item $\CC$ is monadically stable and co-matching-free.
        \item $\CC$ has bounded $r$-subflipper-rank for every $r\in \N$.
    \end{enumerate}
\end{theorem}

Before we prove the forward and backward direction of \Cref{thm:subflipper} in \Cref{sec:subflipper-forward,sec:subflipper-backward},
we first collect some useful facts about the (sub)flipper-rank.

\begin{lemma}\label{lem:flipper-basics}
    For all graphs $G$ and $G'$:
    \begin{enumerate}
        \item $\frk_{r,k}(G) \leq \sfrk_{r,k}(G)$.
        \item If $G'$ is an induced subgraph of $G$, then $\frk_{r,k}(G') \leq \frk_{r,k}(G)$.
        \item If $G'$ is a $k'$-flip of $G$, then $\frk_{r,k\cdot k'}(G') \leq \frk_{r,k}(G)$.
    \end{enumerate}
\end{lemma}
\begin{proof}
    We prove the statements in order.
    \begin{enumerate}
        \item By induction on the subflipper-rank using the fact that every $k$-subflip is also a $k$-flip.
        \item By induction on the flipper-rank using hereditariness (\Cref{lem:flip-basics}).
        \item     By symmetry (\Cref{lem:flip-basics}), $G$ is a $k'$-flip of $G'$.
    Let $H$ be the $k$-flip of $G$ witnessing $\frk_{r,k}(G)\leq d$.
    By transitivity (\Cref{lem:flip-basics}), $H$ is also a $(k \cdot k')$-flip of $G'$ that witnesses $\frk_{r,k\cdot k'}(G')\leq d$.\qedhere
    \end{enumerate}
\end{proof}

We will later prove the bound on the flipper rank by induction and use that at any stage of the induction, our graph is similar to a graph of small co-matching-index. 

\begin{definition}[$\MM$-similarity]
Two graphs $G_1$ and $G_2$ on the same vertex set $V$ are called \emph{$\MM$-similar} if $\MM$ is a partition of $V$ such that $G_1 \ominus \MM = G_2 \ominus \MM$.
We say $G_1$ and $G_2$ are $m$-similar if they are $\MM$-similar for some partition $\MM$ of size at most $m$, which we call a \emph{mediator partition}, or simply a \emph{mediator}.
\end{definition}

\begin{lemma}\label{lem:similarity}
    For all graphs $G$ and $H$, every partition $\MM$ of $V(G)$, and set $S\subseteq V(G)$:
    \begin{enumerate}
        \item $G$ and $G$ are $1$-similar.
        \item $G$ and $G\ominus \MM$ are $\MM$-similar.
        \item If $G$ and $H$ are $\MM$-similar, then they are $\MM$-flips of each other.
        \item If $G$ and $H$ are $\MM$-similar, then $G[S]$ and $H[S]$ are $\MM|_S$-similar.
    \end{enumerate}
    
\end{lemma}
\begin{proof} We prove the statements in order.
    \begin{enumerate}
        \item Consider the mediator $\{ V(G) \}$.
        \item Consider the mediator $\MM$.
        \item $G \ominus \MM$ is an $\MM$-flip of $G$.
    By symmetry of flips (\Cref{lem:flip-basics}), $H$ is an $\MM$-flip of $G \ominus \MM = H \ominus \MM$.
    By transitivity of flips, $H$ is an $(\MM \wedge \MM = \MM)$-flip of $G$.
    \item Applying \Cref{lem:subflip-induce} twice, we get
        \[
         G[S] \ominus \MM|_S 
         = 
         ((G \ominus \MM) [S]) \ominus \MM|_S
         =
         ((H \ominus \MM) [S]) \ominus \MM|_S
         =
         H[S] \ominus \MM|_S 
         .\qedhere
    \]
    \end{enumerate}
\end{proof}

\subsection{Bounding the subflipper-rank}\label{sec:subflipper-forward}

The forward direction of \Cref{thm:subflipper} is implied by the following lemma.

\begin{lemma}\label{lem:subflipper-forward}
    If a graph class $\CC$ has bounded $3tr$-flipper-rank and co-matching-index less than $t$, then $\CC$ also has bounded $r$-subflipper-rank.

    \smallskip\noindent
    More precisely, there exists a function $f : \N^3 \to \N$ such that for every $r,t,k \geq 1$, $d\geq 0$, and graph~$G$ with co-matching-index less than $t$
    \[
        \frk_{3tr,k}(G) \leq d
        \Rightarrow
        \sfrk_{r,k'}(G) \leq d,
    \]
    where $k' := f(t,k,d)$.
\end{lemma}

We will instead prove the following lemma that implies \Cref{lem:subflipper-forward} by choosing $G_0 = G$.
We prove it by approximating the moves of the flipper-game strategy with subflips using \Cref{lem:transfer}.

\begin{lemma}\label{lem:subflipper-inductive}
    There exists a function $f : \N^4 \to \N$ such that for every $r,t,m,k \geq 1$, $d\geq 0$, and graph $G$ that is $m$-similar to a graph~$G_0$ with co-matching-index less than $t$
    \[
        \frk_{3tr,k}(G) \leq d
        \Rightarrow
        \sfrk_{r,k'}(G) \leq d,
    \]
    where $k' := f(t,m,k,d)$.
\end{lemma}

\begin{figure}[htbp]
\begin{center}
\resizebox{\textwidth}{!}{
\begin{tikzcd}
	{G_0} \\
	\\
	& M &&& {H^*} &&&& {H^*[B]} &&&& {H^*[B^*]} \\
	\\
	G &&&& H &&&& {H[B]}
	\arrow["{\ominus \cal M}"{description}, from=1-1, to=3-2]
	\arrow["{\ominus \cal Q}"{description}, from=1-1, to=3-5]
	\arrow["{\ominus \cal Q}"{description}, from=3-2, to=3-5]
	\arrow["{\text{induce on } B}", dashed, from=3-5, to=3-9]
	\arrow["{\text{induce on } B^*}", shift left=3, curve={height=-30pt}, dashed, from=3-5, to=3-13]
	\arrow["{\text{induce on } B^*}", dashed, from=3-9, to=3-13]
	\arrow["{\oplus \cal M}"{description}, tail reversed, from=5-1, to=1-1]
	\arrow["{\ominus \cal M}"{description}, from=5-1, to=3-2]
	\arrow["{\ominus \cal Q}"{description}, from=5-1, to=3-5]
	\arrow["{\oplus \cal P}"{description}, tail reversed, from=5-1, to=5-5]
	\arrow["{\oplus \cal Q}"{description}, tail reversed, from=5-5, to=3-5]
	\arrow["{\text{induce on } B}", dashed, from=5-5, to=5-9]
	\arrow["{\oplus \mathcal{Q}|_B}"{description}, tail reversed, from=5-9, to=3-9]
\end{tikzcd}
}
\end{center}
\caption{An illustration of the proof of \Cref{lem:subflipper-inductive}. A directed solid arrow from a graph $G_1$ to a graph $G_2$ labeled by $\ominus \PP$ means $G_1 \ominus \PP = G_2$.
A bidirectional solid arrow between $G_1$ and $G_2$ labeled by $\oplus \PP$ means $G_1$ and $G_2$ are $\PP$-flips of each other.
A dashed arrow from $G_1$ to $G_2$ means $G_2$ is an induced subgraph of $G_1$.
}
\label{fig:flipper-proof2}
\end{figure}
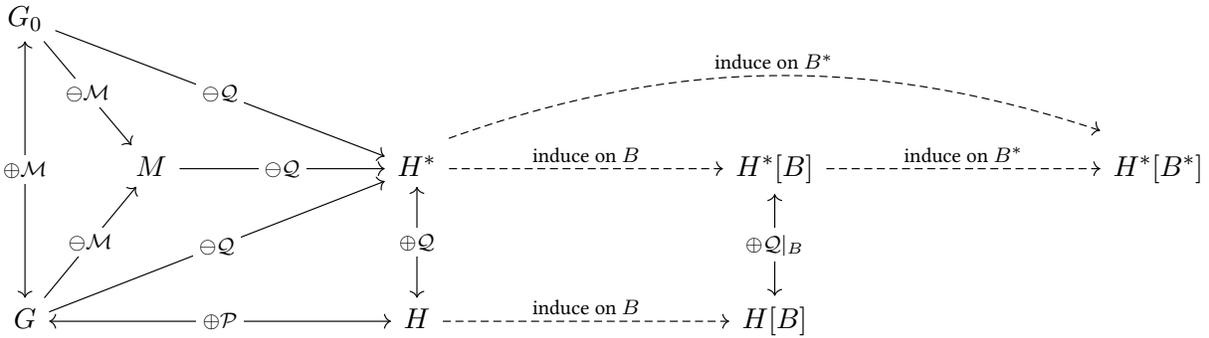

\begin{proof}
    We inductively define 
    \[
        f(t,m,k,0) := m
        \quad
        \text{and}
        \quad
        f(t,m,k,d+1) := f(t,m \cdot k \cdot t^{m \cdot k}, m \cdot k^2 \cdot t^{m \cdot k}, d),
    \]
    and show that $f$ satisfies the lemma by induction on $d$.
    The case $d = 0$ is trivial.
    The inductive step is illustrated in \Cref{fig:flipper-proof2}.
    Assume we are given $G$ and $G_0$ as in the statement of the lemma and~$\MM$ is a mediator partition of size at most $m$.
    By \Cref{lem:similarity}, $G$ and $G_0$ are $\MM$-flips of each other.

    By assumption on the flipper rank of $G$, there is a $\PP$-flip $H$ of $G$ with
    \begin{equation}\label{eq:small-frk-balls}
        \max_{v \in V(H)} \frk_{3tr,k}(\Ball_{3tr}^H(v)) \leq d
    \end{equation}
    for some partition $\PP$ of $V(G)$ of size at most $k$.
    By transitivity, $H$ and $G_0$ are $(\MM \wedge \PP)$-flips of each other. As $G_0$ has co-matching index less than $t$, we can apply \Cref{lem:transfer} to $G_0$, $(\MM \wedge \PP)$, and $H$, which yields a partition $\QQ$ of size
    \[
        |\QQ| \leq m \cdot k \cdot t^{m \cdot k},
    \]
    which refines $(\MM \wedge \PP)$ such that
    $H^* := G_0 \ominus \QQ$ approximates $H$ as follows:
    \begin{equation}\label{eq:subballs}
        \text{for every $v \in V(G)$:} 
        \quad
        N_r^{H^*}[v] \subseteq N_{3tr}^H[v].
    \end{equation}
    Using \Cref{lem:refine} and the fact that $\QQ$ refines $\MM$, we obtain
    \[
        G \ominus \QQ = M \ominus \QQ = G_0 \ominus \QQ = H^*.
    \]
    By transitivity, $H^*$ and $H$ are $(\PP \wedge \QQ = \QQ)$-flips of each other.
    Consider now any vertex $v$, its $3tr$-neighborhood in $H$ and its $r$-neighborhood in $H^*$:
    \[
        B := N_{3tr}^H[v] 
        \quad\text{and}\quad
        B^* := N_{r}^{H^*}[v].
    \]
    By hereditariness of flips (\Cref{lem:flip-basics}), $H[B]$ and $H^*[B^*]$ are $\QQ|_B$-flips of each other.
    Moreover by \Cref{eq:subballs}, $H^*[B^*]$ is an induced subgraph of $H^*[B]$.
    This means we can bound
    \[
    \frk_{3tr,k \cdot |\QQ|}(H^*[B^*]) 
    \leq 
    \frk_{3tr,k \cdot |\QQ|}(H^*[B]) 
    \leq 
    \frk_{3tr,k}(H[B]) 
    \leq 
    d,
    \]
    where the first two inequalities follow by \Cref{lem:flipper-basics}, and the last one by \Cref{eq:small-frk-balls}.
    Additionally by \Cref{lem:similarity}, $H^*[B^*]$ is $|\QQ|_B$-similar to the graph $G_0[B]$, which has co-matching-index less than $t$.
    Therefore, applying induction on $H^*[B^*]$ yields
    \[
                \sfrk_{r,k'}(H^*[B^*])\leq d,
    \]
    where 
    $k' := 
    f(t,|\QQ|,k \cdot |\QQ|,d) = 
    f(t,m \cdot k \cdot t^{m \cdot k}, m \cdot k^2 \cdot t^{m \cdot k}, d) = 
    f(t,m,k,d+1)$.
    We summarize: there is a $(|\QQ| \leq m \cdot k \cdot t^{m \cdot k} \leq k')$-subflip $H^*$ of $G$ in which every $r$-ball 
    $H^*[B^*]$ satisfies $\sfrk_{r,c}(H^*[B^*])\leq d$.
    This means $\sfrk_{r,k'}(G)\leq d + 1$, which completes the proof.
\end{proof}

This finishes the proof of the forward direction of \Cref{thm:subflipper}.

\subsection{Large co-matchings imply large rank}\label{sec:subflipper-backward}
In this subsection we prove the backward direction of \Cref{thm:subflipper}.
We have to show that if a graph class~$\CC$ is either not monadically stable or not co-matching-free, then there is some $r\in\N$ such that $\CC$ has unbounded $r$-subflipper-rank.
By \Cref{thm:flipper}, if a graph class $\CC$ is not monadically stable, then there is some $r\in\N$ such that $\CC$ has unbounded $r$-flipper-rank.
By \Cref{lem:flipper-basics}, $\CC$ has unbounded $r$-subflipper-rank.
This leaves us with the case where $\CC$ is not co-matching-free.

\begin{lemma}
    For every $k \geq 2$, $d \geq 0$, and graph $G$:
    If $G$ has co-matching-index at least $(k^2)^d$, then $\sfrk_{3,k}(G) \geq d$.
\end{lemma}
\begin{proof}
    We show the lemma by induction on $d$.
    The cases $d = 0$ and $d = 1$ are trivial.
    For the inductive step let $G$ be a graph with co-matching-index at least $(k^2)^{d+1}$ for some $d \geq 1$.
    Let $H$ be any $k$-subflip of $G$.
    By \Cref{lem:preserve-cmi}, $H$ contains a semi-induced co-matching-index of size at least~$(k^2)^d$. 
    Consider any vertex $v$ from this co-matching.
    As $(k^2)^d \geq 3$, it is easy to check, that $\Ball_3^H(v)$ contains the entire semi-induced co-matching.
    By induction, we have $\sfrk_{3,k}(\Ball_3^H(v)) \geq d$.
    As $H$ was an arbitrary $k$-subflip of $G$, this proves $\sfrk_{3,k}(G) \geq d+1$.
\end{proof}

\begin{corollary}\label{cor:high-rank}
    Every class $\CC$ that is not co-matching-free has unbounded $3$-subflipper-rank.
\end{corollary}

This finishes the proof of \Cref{thm:subflipper}.
\section{Forbidden induced subflips}\label{sec:forbidden-subflips}

The following characterization of half-graph-free monadically dependent classes (i.e., monadically stable classes) by forbidden induced subgraphs was originally proved in \cite{dreier2024stablemc}. 
We refer to \cite{maehlmann-thesis, pilipczuk2025graphclasseslenslogic} for the following formulation of the result. 
The definition of a \emph{star/clique $r$-crossing} is given below.

\begin{theorem}[\cite{dreier2024stablemc}]\label{thm:forbidden-flips-stability}
For every half-graph-free graph class $\Cc$ the following are equivalent.
\begin{enumerate}
    \item $\Cc$ is monadically dependent. 
    \item For every $r,k \geq 1$ there exists $t \in \N$ such that 
    \begin{itemize}
        \item the star $r$-crossing of order $t$ is not a $k$-flip of any induced subgraph in $\CC$, and 
        \item the clique $r$-crossing of order $t$ is not a $k$-flip of any induced subgraph in $\CC$.
    \end{itemize}
\end{enumerate}
\end{theorem}

Similar to the definition of subflip-flatness, we obtain our characterization by replacing flips with subflips.

\begin{theorem}\label{thm:forbidden-subflips}
For every semi-ladder-free graph class $\Cc$ the following are equivalent.
\begin{enumerate}
    \item $\Cc$ is monadically dependent. 
    \item For every $r,k \geq 1$ there exists $t \in \N$ such that
    \begin{itemize}
        \item the star $r$-crossing of order $t$ is not a $k$-subflip of any induced subgraph in $\CC$, and
        \item the clique $r$-crossing of order $t$ is not a $k$-subflip of any induced subgraph in $\CC$.
    \end{itemize}
\end{enumerate}
\end{theorem}

For~$r \ge 1$, the \emph{star~$r$-crossing} of order~$t$ is the~$r$-subdivision of~$K_{t,t}$ (the biclique of order~$t$).
More precisely, it consists of \emph{roots}~$a_1,\dots,a_t$ and~$b_1,\dots,b_t$
together with~$t^2$ many pairwise vertex-disjoint~$r$-vertex paths~$\{ \pi_{i,j} : i,j \in [t] \}$, whose endpoints we denote as~$\Start(\pi_{i,j})$ and~$\End(\pi_{i,j})$.
Each root~$a_i$ is adjacent to~$\{ \Start(\pi_{i,j}) : j \in [t] \}$,
and each root~$b_j$ is adjacent to~$\{ \End(\pi_{i,j}) : i \in [t] \}$.
See \Cref{fig:patterns}.
The \emph{clique~$r$-crossing} of order~$t$ is the graph obtained from the star~$r$-crossing of order~$t$
by turning the neighborhood of each root into a clique.
It will later be useful to partition the vertices of the two patterns into \emph{layers}~$\LL = \{L_0, \ldots, L_{r+1}\}$:
The 0th layer consists of the vertices~$\{a_1,\dots,a_t\}$.
The~$l$th layer, for~$l \in [r]$, consists of the~$l$th vertices of the paths~$\{ \pi_{i,j} : i,j \in [t] \}$.
Finally, the~$(r+1)$th layer consists of the vertices~$\{b_1,\dots,b_n\}$. 

\begin{figure}[htbp]
    \centering
    \includegraphics{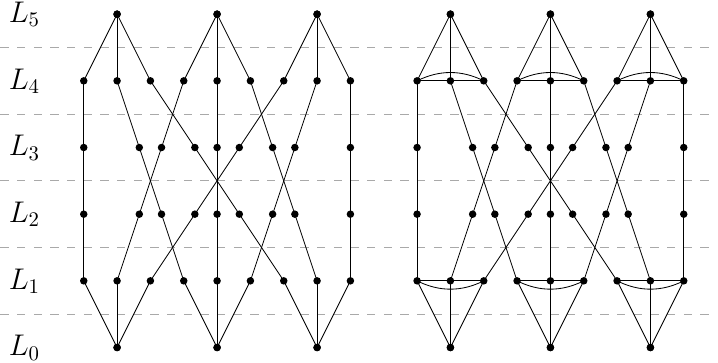}
    \caption{The star $4$-crossing and the clique $4$ crossing of order $3$.}
    \label{fig:patterns}
\end{figure}

We are now ready to prove \Cref{thm:forbidden-subflips}

\begin{proof}[Proof of \Cref{thm:forbidden-subflips}]
    Fix a semi-ladder-free graph class $\CC$.
    The direction $(1. \Rightarrow 2.)$ of the theorem follows immediately from \Cref{thm:forbidden-flips-stability}, as $\CC$ is in particular half-graph-free and every $k$-subflip is also a $k$-flip.

    We prove the direction $(1. \Rightarrow 2.)$ by contraposition. Assume $\CC$ is not monadically dependent.
    We have to show that for some $r,k \in \N$, arbitrary large star $r$-crossings or clique $r$-crossings can be obtained by $k$-subflips from induced subgraphs in $\CC$.
    By \Cref{thm:forbidden-flips-stability}, we know already that we can obtain these patterns from $\CC$ using $k$-flips.
    We will use co-matching-freeness of $\CC$ to argue that subflips are sufficient to extract them.

    Assume first that $\CC$ contains for every $t$, an induced graph $G_t$ with 
    $C_t = G_t \oplus (\PP, F)$ for some size $k$ partition $\PP$, where $C_t$ is the clique $r$-crossings of size $t$.
    By the bipartite structure of the patterns and Ramsey's Theorem, we can assume that $\PP$ is the partition of $V(C_t) = V(G_t)$ into the $r+2$ layers $L_0, \ldots,L_{r+1}$ as in the definition of a clique $r$-crossing (see \cite{dreier2024flipbreakability} or \cite{maehlmann-thesis} for more details on how the Ramseying of these patterns is done).
    Moreover, by the pigeonhole principle we can assume that for every $t\in \N$ the layers are flipped the same way: we use the same $F \subseteq \PP \times \PP$ for every $t \in \N$.
    We want to prove $C_t = G_t \ominus \PP$.
    This means we have to argue that the pairs in $F$ are exactly the pairs of layers that are fully adjacent in $G_t$.
    First note that no two pairs of layers are fully adjacent in~$C_t$. Hence, if there is a pair of fully adjacent layers in $G_t$, it has to be included in~$F$.
    It remains to prove that every pair of layers $(L_i,L_j)\in F$ is fully adjacent in~$G_t$.
    
    First consider the case where $L_i = L_j$, i.e., a layer was already flipped with itself.
    If $i \in \{1,r\}$, then $L_i$ is one of the ``clique layers'': $C_t[L_i]$ is the disjoint union of $t$ many cliques and in particular contains an induced matching of size $t$.
    Then the complement $G_t[L_i]$ contains a semi-induced co-matching of size $t$. Recall that $t$ can be arbitrarily large, so we get a contradiction to $\CC$ being co-matching-free.
    If $i \notin \{1,r\}$, then $G_t[L_i]$ is an independent set so $G_t[L_i]$ is a clique, and $L_i$ is fully adjacent with itself in $G_t$ as desired.

    Assume now $L_i \neq L_j$. If $i$ and $j$ are not successive, then the graph semi-induced between~$L_i$ and~$L_j$ in~$C_t$ is empty, so both sets must be fully adjacent in $G_t$, and we are done.
    For successive layers~$L_i$ and~$L_{i+1} = L_j$ the graph semi-induced in $C_t$ contains an  induced matching of size at least $t$ (size $t$ if an outermost layer is included: $\{i,i+1\} \cap \{0,r+1\}\neq \emptyset$; and size $t^2$ otherwise). 
    As argued before, $G_t$ must contain a semi-induced co-matching of size $t$, and we get a contradiction. 

    This means $C_t = G_t \ominus \PP$, and we can extract arbitrarily large clique $r$-crossings from induced subgraphs of $\CC$ by $(r+2)$-subflips.

    The case where $\CC$ contains flips of arbitrarily large star $r$-crossings $S_t$ works the same way, but needs to consider one less case, as every layer $L_i$ is an independent set in $S_t$.
\end{proof}

As a corollary, we get the following.
\begin{corollary}\label{cor:expos-transduce-everything}
    For every semi-ladder-free class $\CC$ of partially reflexive graphs, the following are equivalent.
    \begin{enumerate}
        \item $\CC$ is monadically dependent.
        \item $\CC$ does not \expos-transduce the class of all reflexive graphs.
    \end{enumerate}
\end{corollary}
\begin{proof}
    If $\CC$ is monadically dependent, then it does not transduce the class of all graphs, so in particular it does not \expos-transduce the class of all reflexive graphs.
    Assume now $\CC$ is not monadically dependent.
    By \Cref{thm:forbidden-subflips}, there are $r,k \geq 1$ such that $\CC$ contains all star $r$-crossings or all clique $r$-crossings as $k$-subflips.

    Assume first it $\CC$ contains all star $r$-crossings as subflips.
    By first taking induced subgraphs and adding all self-loops and then applying \Cref{lem:trans-pure-flips-reflexive}, we can \expos-transduce the class of all reflexive star $r$-crossings from $\CC$.
    For every reflexive graph $G$, the class of all star reflexive $r$-crossings contains the $(2r+1)$-subdivision of $G$ as an induced subgraph.
    From the subdivision, we can transduce $G$ using the \expos-formula that expresses that two vertices are at distance at most $2r+2$ from each other. (This naturally adds self-loops, which is why we required $G$ to be reflexive.)
    This shows we can \expos-transduce the class of all reflexive graphs from the class of all star $r$-crossings, and thus also from $\CC$.

    Assume now $\CC$ contains all clique $r$-crossings as subflips.
    Again $\CC$ \expos-transduces the class of all reflexive clique $r$-crosses.
    From there, we can \expos-transduce all reflexive star $r$-crossings (and thereby reduce to the previous case) by coloring the layers of the crossing and only keeping edges between distinct layers, using the formula
    \[
        \phi(x,y) := \bigvee_{i \neq j \in \{ 0,\ldots,r+1 \}} E(x,y) \wedge L_i(x) \wedge L_j(y).\qedhere
    \]
\end{proof}
\section{A normal form for existential positive FO}\label{sec:normal-form}
\newcommand{\DS}{\mathrm{DS}}
\newcommand{\DC}{\mathrm{DC}}

In this section we establish a normal form for \expos-formulas that we will use for our further arguments. 
We remark that this normal form can also be established for general relational $\Sigma$-structures by considering distances in the Gaifman graph of $G$, we refrain from doing so for consistency of presentation. 

\begin{definition}[Clique formulas]
A formula $\phi(\bar x)$ is an \emph{$r$-clique formula} if for every graph $G$ and every tuple $\bar v \in V(G)^{|\bar x|}$
\[
    G \models \phi(\bar v)
    \quad
    \implies
    \quad
    G \models \bigwedge_{\mathclap{1 \leq  i < j \leq |\bar x|}}\dist(v_i,v_j)\leq r.
\]
This means that the (instantiations of the) free variables of $\phi$ form a \emph{distance-$r$ clique} in $G$.
\end{definition}



We prove the following normal form. 

\begin{theorem}\label{thm:normal-form}
    Every \expos-formula $\phi(\bar x)$ with quantifier rank $q$ is equivalent to an \expos-formula that is a boolean combination of $2^q$-clique formulas $\chi_i(\bar z_i)$
    of quantifier rank at most $q$, where the $\bar z_i$ and are subtuples of $\bar x$.
\end{theorem}


\begin{proof}
    We prove the statement by induction on the structure of $\phi$. 
    The statement is obviously true for atomic formulas $E(x,y)$, $P(x)$ for unary predicates $P$, and for $x=y$. 
    It trivially extends to positive boolean combinations. 
    
    The interesting case is the existential quantifier.
    Assume $\phi(\bar x)=\exists y\, \psi(\bar x,y)$ is a formula of quantifier rank $q \geq 1$.
    Let $r := 2^{q-1}$ and $t := q + |\bar x|$.
    The goal is to show that $\phi$ is equivalent to a positive boolean combination of $2r$-clique formulas.
    
    By applying induction, we can assume that~$\psi(\bar x,y)$ is a positive boolean combination of $r$-clique formulas. 
    We bring~$\psi$ into disjunctive normal form. 
    Note that for this we only have to distribute conjunctions over disjunctions by the equivalence 
    \[
        \zeta_1 \wedge (\zeta_2 \vee \zeta_3) \;\equiv\; (\zeta_1 \wedge \zeta_2) \;\vee\; (\zeta_1 \wedge \zeta_3),
    \]
    hence positivity is preserved by this operation.
    By the equivalence 
    \[
        \exists y \, (\zeta_1(y) \vee \zeta_2(y)) \;\;\equiv\;\; (\exists y \, \zeta_1(y)) \;\vee\; (\exists y \, \zeta_2(y)),    
    \]
    we push the existential quantifier into the disjunction.
    We hence have to deal with the case 
    \begin{align*}
        \phi(\bar x)=\exists y \big( & \chi_1(\bar z_1) \wedge \ldots \wedge \chi_m(\bar z_m)\big),
    \end{align*}
    where the $\chi_i$ are $r$-clique formulas, and the $\bar z_i$ and $\bar z'_i$ are subtuples of $\bar x y$. 
    
    Observe that if the tuple $\bar z_i$ does not contain the quantified variable $y$, then we can move the conjunct~$\chi_i(\bar z_i)$ out of the scope of the quantifier.
    We can therefore assume that every tuple~$\bar z_i$ contains the quantified variable $y$.
    Up to restricting the free variables of $\phi$, we can also assume that each variable from $\bar x$ appears in some tuple $\bar z_i$.

    We claim that $\phi$ is an $r$-clique formula.
    Assume that $\phi(\bar x)$ holds.
    This means there exists an element~$y$ such that all the $r$-clique formulas~$\chi_i(\bar z_i)$ are satisfied.
    Consider now any two variables $x$ and $x'$ from~$\bar x$.
    Let~$\bar z_i$ be a tuple containing $x$ and let $\bar z_j$ be a tuple containing $x'$.
    The truth of the $r$-clique formula~$\chi_i(\bar z_i)$ witnesses that $\dist(x,y) \leq r$ (because $\{x,y\} \subseteq \bar z_i$).
    Similarly, we deduce $\dist(x',y) \leq r$, and by the triangle inequality $\dist(x,x') \leq 2r$, as desired. 
\end{proof}

\section{The monadic non-equality-property}\label{sec:nep}

\begin{definition}\label{def:nep}
    Let $\phi(\bar x, \bar y)$ be an FO formula over the signature of colored graphs, where $\bar x$ and $\bar y$ are two tuple variables of the same length~$t$.
We say that $\phi$ has the \emph{monadic non-equality-property} on a graph class~$\CC$ if for every $\ell \in \N$ there is a graph coloring $G_\ell^+$ of a graph $G_\ell \in \CC$ and tuples $\bar a_1, \ldots, \bar a_\ell \in V(G)^t$ such that for all $i,j\in [\ell]$
\begin{equation*}
    G_\ell^+ \models \phi(\bar a_i, \bar a_j) \iff i \neq j.
\end{equation*}
\end{definition}

The goal of this section is to prove \Cref{thm:nep} from the introduction, which we restate for here for convenience.

\thmNEP*

As reflexive graphs contain no irreflexive cliques, simpler version for reflexive graphs (\Cref{thm:nep-simple}) follows immediately.

In the next two subsections we prove both implications of \Cref{thm:nep} separately.

\subsection{Non-structure}

The property of not having the monadic non-equality-property for all \expos-formulas is preserved by \expos-transductions, as made precise by the following observation. 

\begin{observation}\label{obs:transfer-non-equality}
    For every \expos-formula $\phi(x,y)$ and \expos-transduction $\trans{T}$ there exists an \expos-formula~$\psi(x,y)$ with the following property.
    For every graph class $\CC$: If $\phi$ has the monadic non-equality-property on $\trans{T}(\CC)$, then $\psi$ has the monadic non-equality-property on $\CC$.
\end{observation}

\begin{lemma}[{$(2. \Rightarrow 1.)$ of \Cref{thm:nep}}]\label{lem:nep-non-structure}
    For every graph class $\CC$, if $\CC$ is
    \begin{enumerate}
        \item not irreflexive-clique-free, or
        \item not semi-ladder-free, or
        \item not monadically dependent,
    \end{enumerate}
    then there is an EP-formula with the monadic non-equality property on $\CC$.
\end{lemma}
\begin{proof}
    In the first case, $E(x,y)$ has the monadic inequality property on $\CC$.

    In the second case, $\CC$ \expos-transduces either the class of all co-matchings or the class of all half-graphs.
    By \Cref{obs:transfer-non-equality}, it suffices to show that there are \expos-formulas $\phi$ and $\psi$ with the monadic non-equality-property on the class of all co-matchings and on the class of all half-graphs, respectively.
    We can pick
    \[
        \phi(x_1x_2,y_1y_2) := E(x_1,y_2)
        \text{ and }
        \psi(x_1x_2,y_1y_2) := E(x_1,y_2) \vee E(y_1,x_2).
    \]
    See \Cref{fig:monadic-non-eq} on how to choose witnessing tuples.

    In the third case, if $\CC$ is semi-ladder-free, but not monadically dependent then it \expos-transduces the class of all reflexive graphs by \Cref{cor:expos-transduce-everything}, so we conclude as in the second case.
\end{proof}

\begin{figure}[htbp]
    \centering
    \includegraphics{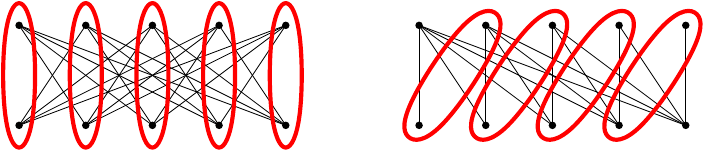}
    \caption{Tuples witnessing the monadic non-equality-property in co-matchings and half-graphs.}
    \label{fig:monadic-non-eq}
\end{figure}

\subsection{Structure}

If additionally in the definition of the monadic non-equality-property (\Cref{def:nep}), we demand that the tuples $\bar a_1,\ldots,\bar a_\ell$ are pairwise disjoint, then we say $\phi$ has the \emph{disjoint} monadic non-equality-property on $\CC$.
Let us show that we can always assume disjointness.

\begin{lemma}\label{lem:assume-disjointness}
    For every graph class $\CC$: if there is an (existential) positive formula $\phi$ that has the monadic non-equality property on $\CC$, then there is also an (existential) positive formula $\psi$ that has the disjoint monadic non-equality property on $\CC$.
\end{lemma}

To prove this lemma we require the following lemma, which can be seen as a sunflower lemma for tuples.
For a $t$-tuple $\bar a = a_1 \ldots a_t$ we write $\bar a[i] := a_i$ for its projection to the $i$th coordinate, and 
$\bar a |_I := a_{i_1} \ldots a_{i_p}$ for the subtuple induced by the set $I = \{i_1 < \ldots < i_p\} \subseteq [t]$.

\begin{lemma}\label{lem:sunflower}
    For all $t,m \in \N$ there is $n\in\N$ such that every set of $n$ many $t$-tuples $S$, there is a subset $A = \{ \bar a_1, \ldots, \bar a_m \}$ of $m$ many $t$ tuples and a set $I \subseteq [t]$ such that for all distinct $\bar a, \bar a' \in A$ and $J := [t] \setminus I$, the subtuples $\bar a |_I$ and $\bar a'|_I$ are equal, and the subtuples $\bar a |_J$ and $\bar a'|_J$ are disjoint.
\end{lemma}
\begin{proof}
    Fix $t,m\in\mathbb N$ and let $n$ be large enough for the classical sunflower lemma for set systems of size at most~$t$ and sunflower size $m$
    (recall that a family of sets forms a sunflower with core $C$ if all pairwise intersections equal $C$, and the sunflower lemma guarantees the existence of big sunflowers in large set systems).
    Let $S$ be a set of $n$ many $t$-tuples.
    For a tuple $\bar a=(a_1,\ldots,a_t)$ define its support $\mathrm{supp}(\bar a):=\{a_1,\ldots,a_t\}$ (the tuple interpreted as a set). 
    Applying the sunflower lemma to the family $\{\mathrm{supp}(\bar a) : \bar a\in S\}$, we obtain a subset $A=\{\bar a_1,\ldots,\bar a_m\}\subseteq S$ whose supports form a sunflower with core~$C$.
    That is, for all $i\neq j$, $\mathrm{supp}(\bar a_i)\cap \mathrm{supp}(\bar a_j)=C$,
    and the sets $\mathrm{supp}(\bar a_i)\setminus C$ are pairwise disjoint.

    For each $\bar a\in A$, let $I(\bar a):=\{k\in[t] : \bar a[k]\in C\}$.
    Since there are only finitely many subsets of $[t]$, by the pigeonhole principle we may assume (after discarding tuples if necessary) that $I(\bar a_1)=\cdots=I(\bar a_m)=:I$. 
    Since there is only a constant number of ways to assign indices from $I$ to the set $C$, by applying the pigeonhole principle again, we can assume that $\bar a_i|_I=\bar a_j|_I$ for all $i,j$.
    On the other hand, for $k\in J:=[t]\setminus I$ we have $\bar a_i[k]\in \mathrm{supp}(\bar a_i)\setminus C$, and since the petals of the sunflower are pairwise disjoint, the subtuples $\bar a_i|_J$ and $\bar a_j|_J$ are disjoint for all~$i\neq j$.
\end{proof}

We can now prove \Cref{lem:assume-disjointness}.

\begin{proof}[Proof of \Cref{lem:assume-disjointness}]
    Let $\phi(\bar x, \bar y)$ be the formula that has the monadic non-equality-property on~$\CC$ for tuples of length $t := |\bar x| = |\bar y|$.
    By \Cref{lem:sunflower} and the pigeonhole principle, we can assume that there is a single set $I \subseteq [t]$ such that for every $\ell \in \N$ there is a coloring $G^+_\ell$ of a graph $G_\ell \in \CC$ containing tuples $\bar a_1, \ldots, \bar a_\ell$ witnessing the monadic non-equality-property, such that additionally 
    the $\bar a_i |_I$ subtuples are pairwise equal and the $\bar a_i |_J$ subtuples are pairwise disjoint for $J:= [t] \setminus I$.
    Up to permuting the variables in $\phi$, we can assume that $I = \{p+1,\ldots,t\}$ for some $p \in [t]$. (We know that $I \subsetneq [t]$, as otherwise all tuples would be equal and could not witness the monadic non-equality-property.)
    We add to each $G^+_\ell$ for each $i \in I$ a predicate $X_i$ that marks the single element $\bar a_1[i] = \ldots = \bar a_\ell[i]$.
    Now the (existential) positive formula 
    \[
        \psi(\bar x' = x_1 \ldots x_{p}, \bar y' = y_1 \ldots y_p)
        :=
        \exists \bar z \in (X_{p+1} \times \ldots \times  X_{t})\ \phi(\bar x' \bar z, \bar y' \bar z)
    \]
    witnesses the disjoint monadic non-equality-property on the tuples $\bar a_1 |_J, \ldots, \bar a_\ell |_J$.
\end{proof}

\begin{lemma}[{$(1. \Rightarrow 2.)$ of \Cref{thm:nep}}]
    For every class $\CC$ of partially reflexive graphs, if $\CC$ is
    \begin{enumerate}
        \item irreflexive-clique-free, and
        \item semi-ladder-free, and
        \item monadically dependent,
    \end{enumerate}
    then no EP-formula has the monadic non-equality property on $\CC$.
\end{lemma}
\begin{proof}
    Towards a contradiction, assume that there is an \expos-formula that has the monadic non-equality-property on $\CC$.
    By \Cref{lem:assume-disjointness} there also exists an \expos-formula $\phi(\bar x, \bar y)$ that has the disjoint monadic non-equality-property on $\CC$.
    Fix a large set of pairwise disjoint tuples of vertices $A = \{ a_{1}, \ldots, a_{\ell} \}$ witnessing that $\phi$ has the disjoint monadic non-equality-property in some graph $G_\ell \in \CC$.
    This means there is a coloring $G^+_\ell$ of $G_\ell$ such that for all $i,j \in [\ell]$ 
    \begin{equation*}
        G_\ell^+ \models \phi(\bar a_i, \bar a_j) \Leftrightarrow i \neq j.
    \end{equation*}
    (It is possible but not fruitful to determine the required value of $\ell$ from the proof that follows.
    It is only important that $\ell$ depends only on $\phi$ and $\CC$.)
    Let $q$ be the quantifier rank of $\phi$.
    Since~$\CC$ is semi-ladder-free and monadically dependent, it satisfies the conditions of \Cref{thm:subflatness} and is $2^{q+1}$-subflip-flat.
    By \Cref{lem:sff-tuples}, $\CC$ is also $2^q$-subflip-flat for tuples.
    Then there exists $k\in\N$ such there is a $k$-subflip $H$ of $G_\ell$ in which many of the tuples $A' \subseteq A$ are pairwise at distance larger than $2^q$. 
    By irreflexive-clique-freeness of $\CC$, we can apply \Cref{cor:transduce-subflips-partially-reflexive} and obtain a quantifier-free positive formula $\psi(x,y)$ that only depends on $k$ and $\CC$.
    Replacing the edge relation $E(x,y)$ with $\psi$ in $\phi$ results in a formula $\phi^*$ for which there is a coloring~$H^+$ of~$H$ such that for all $i,j \in [\ell]$
    \begin{equation}\label{eq:nep-sff}
        H^+ \models \phi^*(\bar a_i, \bar a_j) \Leftrightarrow G_\ell^+ \models \phi(\bar a_i, \bar a_j) \Leftrightarrow i \neq j.
    \end{equation}
    Because $\psi$ is quantifier-free and positive, $\phi^*$ is still an \expos-formula of quantifier rank $q$.
    Applying the normal form for \expos-formulas (\Cref{thm:normal-form}) we can rewrite $\phi^*$ into a positive DNF of $2^q$-clique formulas of  quantifier rank at most $q$.
    This means $\phi$ is a disjunction of formulas of the form
    \begin{equation}\label{eq:dnf}        
        \alpha(\bar x, \bar y) = \alpha_1(\bar z_1) \wedge \ldots \wedge \alpha_m(\bar z_m),
    \end{equation}
    where each $\alpha_i$ is a $2^q$-clique formula of quantifier rank at most $q$ and each $\bar z_i$ is a subtuple of $\bar x \bar y$.
    The number of FO formulas with quantifier rank $q$ over the signature of graphs colored with the predicates appearing in $\phi^*$ is bounded (\cref{fact:type-bound}). 
    Since $A'$ is large, it contains two distinct tuples $\bar a$ and $\bar a'$ that satisfy the same quantifier rank $q$ formulas in $H^+$.
    By construction,
    \[
        H^+ \models \phi^*(\bar a, \bar a'),
    \]
    and this is witnessed by a satisfied disjunct $\alpha$ of $\phi^*$ as in \Cref{eq:dnf}.
    This means all the conjuncts~$\alpha_i$ of $\alpha$ are satisfied.
    Since the conjuncts are $2^q$-clique formulas and the tuples $\bar a, \bar a'$ are at distance greater than~$2^q$ from each other in $H^+$, none of them can have in their free variables elements from both of the tuples.
    Hence, the truth of $\alpha$ depends only on the quantifier rank $q$ formulas satisfied by $\bar a$ and $\bar a'$ individually.
    Since they satisfy the same formulas in $H^+$, we conclude that also $H^+ \models \alpha(\bar a, \bar a)$ and thus also $H^+ \models \phi^*(\bar a, \bar a)$; a contradiction to \Cref{eq:nep-sff}.
\end{proof}

As a corollary, we obtain \Cref{thm:co-matching-freeness-preserved} from the introduction.

\thmCoMatchingFreenessPreserved*

\begin{proof}
    As $\CC$ is nowhere dense, it satisfies all conditions of \Cref{thm:nep}.
    Hence, no \expos-formula has the monadic non-equality-property on $\CC$.
    By \Cref{obs:transfer-non-equality}, the latter also holds for every class~$\DD$ that is \expos-transducible from $\CC$.
    Then by \Cref{lem:nep-non-structure}, $\DD$ must be semi-ladder-free.
\end{proof}
\section{\expos-sparsification for bounded subflip-depth}\label{sec:sparsification-into-subgraphs}

\newcommand{\flipdepth}{\mathrm{flip\text{-}depth}}
\newcommand{\subflipdepth}{\mathrm{subflip\text{-}depth}}
\newcommand{\SCdepth}{\mathrm{SC\text{-}depth}}
\newcommand{\treedepth}{\mathrm{tree\text{-}depth}}

As a stepping stone to \Cref{thm:sparsification}, we define in this section the graph parameter \emph{subflip-depth} and show that every class of bounded subflip-depth can be EP-transduced from a class of bounded tree-depth. (And the reverse is true as well by \Cref{thm:co-matching-freeness-preserved}.)

\subsection{Definition and basic properties}
We define \emph{flip-depth} and \emph{subflip-depth} through the (sub)flipper-rank for radius $r = \infty$.
For every graph $G$ and $k\in \N$, we define
\[
    \textrm{flip-depth}_k(G) := \frk_{\infty, k}(G) \quad\text{and}\quad  \textrm{subflip-depth}_k(G) := \sfrk_{\infty, k}(G).
\]
We recall the definition of the flipper-rank for the special case of radius $r = \infty$ here for convenience.

\begin{definition}[(Sub)flip-depth]
    For every $k\in\NN$ we
    define the \emph{$k$-flip-depth} of the single vertex graph~$K_1$~as
    \[
        \flipdepth_k(K_1) := 0.
    \]
    For every other graph $G$ we define
    \[
        \flipdepth_k(G) := 1+ \min_{\substack{\text{$H$ is a} \\ \text{$k$-flip of $G$}}}\quad \max_{C\in \mathrm{Conn}(H)} \flipdepth_k(C),
    \]
    where $\mathrm{Conn}(H)$ are the connected components of $H$.
    We naturally obtain the definition of \emph{$k$-subflip-depth} by replacing $k$-flips with $k$-subflips in the above.
    A graph class $\CC$ has \emph{bounded (sub)flip-depth}, if there are $k,d \in \N$ such that every graph $G \in \CC$ has $k$-(sub)flip-depth at most $d$.
\end{definition}

Better known than $k$-flip-depth is the parameter \emph{SC-depth}~\cite{shrubdepth}, which is defined by replacing $k$-flips with \emph{subset-complements} in the definition of $k$-flip-depth.
We say $H$ is a subset-complement of $G$, if~$H$ can be obtained from~$G$ by complementing the edges on a subset of vertices. In the language of flips:
$H = G \oplus (\{A, V(G)\setminus A\}, \{(A,A)\})$ for some set~$A\subseteq V(G)$.
It is folklore that Flip-depth and SC-depth are functionally equivalent.

\begin{lemma}\label{rem:flipdepth-SCdepth}
    For every graph $G$ and $k \geq 2$
    \[
        \flipdepth_k(G) 
        \leq 
        \SCdepth(G)
        \leq 
        3k^2 \cdot \flipdepth_k(G).
    \]
\end{lemma}
\begin{proof}
    The first inequality follows because every subset-complement is a $2$-flip.
    For the second inequality, notice that flipping between two parts $P$ and $Q$ of a partition can be simulated by performing three subset-complementations: we complement $P\cup Q$, $P$, and $Q$. Then, performing a $k$-flip can be simulated by flipping at most $k^2$ pairs of parts, one after the other.
\end{proof}

Through the above equivalence, the following facts about SC-depth carry over to flip-depth.

\begin{theorem}[{\cite{shrubdepth,shrubdepth-journal}}]\label{lem:sc-depth-no-long-paths}
    For every $k,d$ there is $t\in\N$ such that no graph $G$ with $k$-flip-depth $d$ contains an induced path of length $t$.
\end{theorem}

\begin{theorem}[\cite{dreier2024flipbreakability}]\label{thm:infty-ff}
    A graph class $\CC$ has bounded flip-depth if and only if it is $\infty$-flip-flat.
\end{theorem}

This yields the following characterization of subflip-depth.

\begin{theorem}\label{thm:characterizations-subflip-depth}
    For every graph class $\CC$ the following are equivalent.
    \begin{enumerate}
        \item $\CC$ is co-matching-free and has bounded flip-depth.
        \item $\CC$ has bounded subflip-depth.
        \item $\CC$ is $\infty$-subflip-flat.
    \end{enumerate}
\end{theorem}

\begin{proof}
    We have proven all required implications already.
    \begin{itemize}
        \item $(1. \Rightarrow 2.)$ follows from \Cref{lem:subflipper-forward}.
        \item $(2. \Rightarrow 1.)$ follows from \Cref{cor:high-rank}.
        \item $(1. \Rightarrow 3.)$ follows from \Cref{thm:infty-ff} and \Cref{lem:subflatness-forward}.
        \item $(3. \Rightarrow 1.)$ follows from \Cref{thm:infty-ff} and \Cref{lem:subflatness-backward}.\qedhere
    \end{itemize}
\end{proof}

\subsection{Sparsification}
The goal of this subsection is to prove the following theorem.

\begin{restatable}[\needsselfloops]{theorem}{thmSparsifySfrk}\label{lem:sparsify-sfrk}
    For every $k,d \in \N$ there are \expos-transductions $\trans {Sparsify}_{k,d}$ and $\trans{Recover}_{k,d}$ such that for every reflexive graph $G$ with $k$-subflip-depth at most $d$ there is a reflexive subgraph $G^*$ of $G$ with tree-depth at most $k\cdot d$ such that
    \[
        G^* \in \trans{Sparsify}_{k,d}(G)
        \text{ and } 
        G \in \trans{Recover}_{k,d}(G^*).
    \]
\end{restatable}

We recall the definition of tree-depth.
\begin{definition}[Tree-depth]
    We
    define the \emph{tree-depth} of the single vertex graph $K_1$ as
    \[
        \treedepth(K_1) := 0.
    \]
    For every other graph $G$ we define
    \[
        \treedepth(G) := 1+ \min_{v \in V(G)}\ \ \max_{C\in \mathrm{Conn}(G - v)} \treedepth(C),
    \]
    where $G - v$ is the graph $G$ with $v$ removed and $\mathrm{Conn}(G - v)$ are the connected components of~$G - v$.
\end{definition}

Before we can start sparsifying graphs of bounded subflip-depth, we need some additional tools.
We first show how we can use transductions to add and remove edge sets that have a small vertex cover.
A set of vertices $S \subseteq V$ is a \emph{vertex cover} for a set of edges $X \subseteq V \times V$, if for each $uv \in X$ we have $\{u,v\} \cap S \neq \emptyset$.

\begin{lemma}[\needsselfloops]\label{lem:transduce-vc}
For every $k\in \N$, there is a quantifier-free, positive transduction $\trans{VertexCover}$ such that for every reflexive graph $G = (V,E)$ and edge set $X \subseteq E$ with a vertex cover of size at most $k$
\begin{center}
    $H \in \trans{VertexCover}(G)$ and $G \in \trans{VertexCover}(H)$, where $H := (V, E - X)$.
\end{center}
\end{lemma}

\pagebreak
\begin{proof}
    Let $S = \{v_1, \ldots, v_k\}$ be a vertex cover of $X$.
    For each $i \in [k]$, let $\PP_i$ be the partition of $V(G)$ into 
    \begin{enumerate}
        \item a part $V_i$ containing only $v_i$,
        \item a part $X_i$ containing all vertices $x$ such that $v_i x \in X$,
        \item a part containing the remaining vertices.
    \end{enumerate}
    Let $\PP$ be the common refinement of all the $\PP_i$s. We have $|\PP| \leq 3^k$.
    Let $R$ be the relation that contains for every $i \in [k]$ all pairs $(V_i,X'_i)$ and $(X'_i, V_i)$, where $X'_i$ is a part of $\PP$ that is contained in~$X_i$.
    All pairs in $R$ are fully adjacent in $G$ and $H = G \oplus (\PP,R)$, so we conclude in both directions by applying \Cref{lem:trans-pure-flips-reflexive}.
\end{proof}

Next, we show how to apply a transduction to all connected components of a graph in parallel, in the case where all components have bounded diameter.

\begin{lemma}\label{lem:parallel-transductions}
    For every $d\in \N$ and \expos-transduction $\trans T$, there is a \expos-transduction $\trans{P}$ with the following property.
    For every two partially reflexive graphs $G$ and $G^*$, if they satisfy
    \begin{itemize}
        \item $H_1, \ldots, H_m$ are the connected components of $G$,
        \item every $H_i$ has diameter at most $d$,
        \item $H^*_1, \ldots, H^*_m$ are the connected components of $G^*$, and
        \item every $H^*_i$ is contained in $\trans T(H_i)$,
    \end{itemize}
    then $G^* \in \trans P(G)$.
\end{lemma}

\begin{proof}
    Let $\phi(x,y)$ be the formula specifying $\trans T$.
    We build the formula $\phi^*(x,y)$ that specifies $\trans P$ by 
    \begin{enumerate}
        \item modifying the formula $\phi(x,y)$ to only create edges between two vertices $x$ and $y$ that are in the same connected component of $H$, and
        \item relativizing all quantifiers of $\phi$ to the connected component of $x$.
    \end{enumerate}
    More precisely, we set
    \[
        \phi^*(x,y) := \mathrm{sameComponent}(x,y) \wedge \mathrm{relativize}(\phi)(x,y),
    \]
    where $\mathrm{relativize}(\phi)$ is obtained by recursively replacing each quantified subformula $\exists z\ \psi(\bar z z)$ of $\phi$ with 
    \[
        \exists z\ \mathrm{sameComponent}(x,z) \wedge \psi(\bar z z).    
    \]
    Crucially relying on the diameter bound, the predicate $\mathrm{sameComponent}(x,y)$ is expressible in existential positive FO, as it just checks whether $x$ and $y$ are at distance at most $d$.

    When evaluating $\phi^*$ in the graph $H$ on two vertices $u$ and $v$ in the same connected component~$H_i$, we get the same result as if we evaluated $\psi$ in just the graph $H_i$.
    Hence, we transduce $H^*$.
    It is easy to verify that $\psi^*$ remains an \expos-formula.
\end{proof}

We are now ready to prove \Cref{lem:sparsify-sfrk}, which we restate for convenience.

\thmSparsifySfrk*

\begin{proof}
    We first prove, by induction on $d$, a slight variation of the lemma, where we additionally
    \begin{itemize}
        \item assume that the input $G$ is connected, and
        \item demand that the result $G^*$ is connected.
    \end{itemize}
    For the base case $d=0$, we can choose $G^* := G = K_1$ and $\trans{Sparsify}_{k,1}$ and $\trans{Recover}_{k,1}$ to be the identity transduction.
    For the inductive step, assume $G$ has $k$-subflip-depth at most $d+1$.
    Let $\PP$ be the partition of size at most $k$ such that the connected components $H_1,\ldots,H_m$ of the graph
    \[
        H := H_1 \uplus \ldots \uplus H_m = G \ominus \PP    
    \]
    all have $k$-subflip-depth at most $d$.
    Applying induction to these components, we construct the graph
    \[
        H^* := H_1^* \uplus \ldots \uplus H_m^*
    \]
    We build $G^*$ by adding to $H^*$ the edges $X$ defined as follows.
    Pick an arbitrary \emph{leader} vertex $\ell(P)\in P$ from each part $P\in \PP$.
    Build $X$ by including:
    \begin{itemize}
        \item for each part $P$ that forms a clique in $G$: an edge from $\ell(P)$ to each other vertex of $P$, and
        \item for each pair of distinct parts $P,Q$ that are fully adjacent in $G$: an edge from $\ell(P)$ to each vertex of $Q$ and an edge $\ell(Q)$ to each vertex from $P$.
    \end{itemize}
    
    \begin{observation}
        The set of leaders forms a vertex cover of size at most $k$ for $X$.
    \end{observation}

    We now check that $G^*$ satisfies the statement of the lemma.

    \begin{claim}
        $G^*$ is a connected subgraph of $G$.
    \end{claim}
    \begin{claimproof}
        All the $H_i^*$ are connected subgraphs of $G$ by induction.
        Clearly $G^*$ is a subgraph of $G$: $G^*$ only adds to the $H_i^*$ the edges $X$ edges that were already present in $G$.

        It remains to show that $G^*$ is connected. As $G$ is connected, it suffices to show that for every edge~$uv$ in $G$, there is a path between $u$ and $v$ in $G^*$.
        If $u$ and $v$ are in the same $H_i^*$, this follows by induction.
        Otherwise, $u$ and $v$ are in different components $H_i$ and $H_j$ of $G \ominus \PP$, so the edge $uv$ was removed by applying the subflip.
        Hence, either $u$ and $v$ are in the same part $P \in \PP$ that is a clique in $G$, or they are in two distinct parts $P$ and $Q$ that are fully adjacent in $G$.
        In both cases we can verify that there is a path of length at most $3$ using only edges from $X$ in $G^*$.
    \end{claimproof}

    \begin{claim}
        $G^*$ has tree-depth at most $k \cdot (d+1)$
    \end{claim}
    \begin{claimproof}
        $H^*$ is a disjoint union of graphs of tree-depth at most $k\cdot d$ by induction.
        It follows from the definition of tree-depth that adding the edges $X$ which have a vertex cover of size at most $k$ can increase the tree-depth by at most $k$.
    \end{claimproof}
    
    \begin{claim}
        There is an \expos-transduction $\trans{Sparsify}_{k,d+1}$ that produces $G^*$ from $G$.
    \end{claim}
    \begin{claimproof}
    As subflips are transducible (\Cref{lem:trans-pure-flips-reflexive}), we can transduce $H$ from $G$.
    By induction, $\trans{Sparsify}_{k,d}$ produces $H_i^*$ from $H_i$ for all $i\in[m]$.
    As each of the $H_i$ has bounded diameter (\Cref{lem:sc-depth-no-long-paths}), we can apply $\trans{Sparsify}_{k,d}$ in parallel to transduce $H^*$ from $H$ (\Cref{lem:parallel-transductions}).
    We can transduce $G^*$ from $H^*$ by adding the additional edges $X$ through \Cref{lem:transduce-vc}.
    By transitivity (\Cref{lem:expos-transitive}), we can chain these transductions to transduce $G^*$ from $G$.
    \end{claimproof}

    \begin{claim}
        There is an \expos-transduction $\trans{Recover}_{k,d+1}$ that produces $G$ from $G^*$.
    \end{claim}
    \begin{claimproof} 
    We first transduce $H^*$ from $G^*$ by removing the edges $E(G^*) - E(H^*) = X$ (\Cref{lem:transduce-vc}).
    Recall that $H^* = H_1^* \uplus \ldots \uplus H_m^*$.
    By induction, $\trans{Recover}_{k,d}$ produces $H_i$ from $H_i^*$ for all \mbox{$i\in[m]$}. 
    Crucially, each~$H_i^*$ is connected and has bounded tree-depth and thus bounded diameter (\Cref{lem:sc-depth-no-long-paths}).
    We can therefore apply $\trans{Recover}_{k,d}$ in parallel to all $H_i^*$ (\Cref{lem:parallel-transductions}).
    This yields $H$.
    As the last step, we undo the subflip and obtain $G$ (\Cref{lem:trans-pure-flips-reflexive}).
    \end{claimproof}

    This finishes the proof of the lemma, where we assumed that $G$ is connected.
    By \Cref{lem:parallel-transductions}, the original statement of the lemma now follows.
\end{proof}
\section{\expos-sparsification for structurally bounded expansion}
\label{sec:sbe}

In this section we prove \Cref{thm:sparsification}, which we restate for convenience. 

\thmSparsification*

The key to handling all five cases of \Cref{thm:sparsification} at once is to work with the class property \emph{structurally bounded expansion}.
A graph class has \emph{structurally bounded expansion} if it is transducible (with general FO, not restricted to the \expos-fragment) from a class of \emph{bounded expansion}~\cite{gajarsky2020first}.
Classes of  \emph{bounded expansion}~\cite{bounded_expansion} are very general classes of sparse graphs, subsuming all classes of bounded tree-depth, path-width, tree-width, and sparse-twin-width, but less general than nowhere dense classes.
Classes of structurally bounded expansion subsume the half-graph-free fragments of the dense, transduction-closed counterparts of the previously mentioned sparse classes.

\begin{theorem}\label{thm:sbe-and-transclosed}
    The following class properties $\PP$ are transduction-closed and every half-graph-free class with~$\PP$ has structurally bounded expansion.
    \begin{multicols}{3}
    \begin{enumerate}
        \item shrub-depth~\cite{shrubdepth,shrubdepth-journal}
        \item linear clique-width~\cite{nevsetvril2020linear}
        \item clique-width~\cite{nevsetvril2021rankwidth}
        \item twin-width~\cite{twwI, gajarsky2022stable}, 
        \item merge-width~\cite{merge-width,dreier2026efficient}.
    \end{enumerate}
    \end{multicols}
\end{theorem}

We will not work with the definition of (structurally) bounded expansion, but instead use the following characterizations through \emph{covers}.

\begin{definition}[$p$-covers]
    Fix a universe $V$.
    A collection of sets $U_1, \ldots, U_s \subseteq V$ forms a \emph{$p$-cover} of $V$ if for every subset $X \subseteq V$ of size at most $p$, there is $i\in[s]$ with $X \subseteq U_i$.
\end{definition}

\begin{definition}[$\PP$-covers]
    For a class property $\PP$ and a graph class $\Cc$, we say $\CC$ has \emph{$\PP$-covers} if for every $p\in \N$ there exists $s\in \N$ such that for every $G\in \Cc$ there is a size-$s$ $p$-cover $U_1,\ldots,U_s$ of~$V(G)$ such that the class $\Dd = \{ G[U_i] \mid G\in \Cc, i\in [s]\}$ has property $\PP$.
\end{definition}

\begin{theorem}\label{lem:sbe-low-fd-covers}
    For every graph class $\Cc$, all of the following hold.
    \begin{enumerate}
        \item $\CC$ has bounded expansion if and only if $\CC$ has bounded-tree-depth-covers.~\cite{bounded_expansion}
        \item $\CC$ has structurally bounded expansion if and only if $\CC$ has bounded-SC-depth-covers.~\cite{gajarsky2020first}
    \end{enumerate}
\end{theorem}

We add to this, and later use, the following equivalence.

\pagebreak
\begin{theorem}\label{thm:subflip-depth-covers}
    For every graph class $\CC$, the following are equivalent.
    \begin{enumerate}
        \item $\CC$ is co-matching-free and has structurally bounded expansion.
        \item $\CC$ has bounded-subflip-depth-covers.
    \end{enumerate}
\end{theorem}

In order to prove this theorem it will be useful to observe that co-matching-freeness can be lifted from the cover sets to the whole graph. 

\begin{lemma}\label{lem:co-matching-covers}
    Let $G$ be a graph and $U_1, \ldots, U_s$ be a $2$-cover of $V(G)$.
    If $G$ contains a semi-induced co-matching of size $s \cdot t$, then there is $i \in [s]$ such that $G[U_i]$ contains a semi-induced co-matching of size~$t$.
\end{lemma}
\begin{proof}
    Let $a_1,\ldots,a_{s\cdot t}$ and $b_1, \ldots, b_{s\cdot t}$ be the vertices that semi-induce a co-matching in $G$, i.e., 
    \[
        a_ib_j \in E(G) \Leftrightarrow i \neq j \quad \text{for all $i,j \in [s \cdot t]$.}
    \]
    Using the assumption that the $U_i$ form a $2$-cover, we can assign to each index $m\in[s \cdot t]$ a cover set $f(m) \in [s]$ with $\{ a_m, b_m \} \in U_{f(m)}$.
    By the pigeonhole principle there, exists indices 
    $1 \leq i_1 < \ldots < i_t \leq s \cdot t$ which are all assigned to the same cover set $U_i$ for some $i \in [s]$. 
    Hence, $G[U_i]$ contains a semi-induced co-matching of size $t$.
\end{proof}

We can now prove \Cref{thm:subflip-depth-covers}.

\begin{proof}[Proof of \Cref{thm:subflip-depth-covers}]
    To prove $(1. \Rightarrow 2.)$, if $\CC$ has structurally bounded expansion then it has bounded-SC-depth-covers, by
    \Cref{lem:sbe-low-fd-covers}.
    By \Cref{rem:flipdepth-SCdepth}, it has bounded-flip-depth-covers.
    Using co-matching-freeness, by \Cref{thm:characterizations-subflip-depth}, it has bounded-subflip-depth-covers.

    To prove $(2. \Rightarrow 1.)$, if $\CC$ has bounded-subflip-depth-covers, then it also has bounded-flip-depth-covers by definition.
    By \Cref{rem:flipdepth-SCdepth}, it has bounded-SC-depth-covers.
    By \Cref{lem:sbe-low-fd-covers}, $\CC$ has structurally bounded expansion.
    It remains to prove co-matching-freeness.
    We know that the graphs from $\CC$ can be $2$-covered by graphs from a class $\DD$ of bounded subflip-depth.
    By \Cref{thm:characterizations-subflip-depth}, $\DD$ is co-matching-free. By contraposition of \Cref{lem:co-matching-covers}, co-matching-freeness lifts to $\CC$.
\end{proof}

Our goal to use \Cref{thm:subflip-depth-covers} to sparsify graphs of structurally bounded expansion, by covering it with a bounded number of graphs of bounded subflip-depth.
These we can individually sparsify to graphs of bounded tree-depth using \Cref{lem:sparsify-sfrk} from the previous section.
To control the density of the graph that results from taking the union of the bounded tree-depth graphs, we will use the following Ramsey-type lemma. 

\begin{lemma}\label{lem:no-bicliques}
For all $s,t\in\N$ there exists $f=f(s,t)$ such that the following holds.
For every graph $G = \bigcup_{i \in [s]} H_i$,
if $G$ contains a semi-induced $K_{f,f}$, then there is $i\in[s]$ such that $H_i$ contains a semi-induced~$K_{t,t}$.
\end{lemma}

\begin{proof}
Let $A$ and $B$ be the sets semi-inducing a biclique in $G$, i.e., $A\times B \subseteq E(G)$.
For every edge $ab \in A \times B$ choose an index
$i(ab)\in[s]$ such that $ab \in E(H_i)$.
This yields an $s$-coloring of the edges of the complete bipartite graph
$K_{|A|,|B|}$.

By the bipartite Ramsey theorem there exists $f=f(s,t)$ such that if
$|A|,|B|\ge f$, then there are subsets $A'\subseteq A$ and
$B'\subseteq B$ of size $t$ and an index $i\in[s]$ such that
all edges between $A'$ and $B'$ are colored~$i$. 
This shows that $H_i$ contains a semi-induced $K_{t,t}$. 
\end{proof}

We will use the above lemma in combination with the following corollary of \Cref{lem:sbe-low-fd-covers} to obtain a class of bounded expansion.
\begin{corollary}[\cite{gajarsky2020first}]\label{cor:sparse-sbe}
    A biclique-free graph class has bounded expansion if and only if it has structurally bounded expansion.
\end{corollary}

Now, in classes of bounded expansion, we can recover subgraphs from a union of graphs using the following lemma. 

\begin{lemma}[follows from Lemma 34 of \cite{nesetril2020structuralc}]\label{lem:transduce-subgraphs}
    Let $\Cc$ be a class with bounded expansion. Then there is an \expos-transduction from $\Cc$ onto its monotone closure $\{ H : \text{$H$ is a subgraph of a graph in $\CC$} \}$.
\end{lemma}

It is now time to combine everything and prove the following lemma, which captures the essence of \Cref{thm:sparsification}.

\begin{lemma}\label{lem:sbe-sparsify}
For every co-matching-free class $\CC$ of reflexive graphs of structurally bounded expansion, there are \expos-transductions $\trans{Sparsify}$ and $\trans{Recover}$ such that every graph $G\in\CC$ contains a reflexive subgraph~$G^*$ such that 
\[
    G^* \in \trans{Sparsify}(G) \quad \text{and }\quad G \in \trans{Recover}(G^*),
\]
and the class $\Dd := \{G^* : G\in\Cc\}$ has bounded expansion.
\end{lemma}

\begin{proof}
By \Cref{thm:subflip-depth-covers}, $\CC$ has bounded-subflip-depth-covers: there are $s,k,d \in \N$ such that for every graph $G \in \CC$, there is a size-$s$ $2$-cover $U_1, \ldots, U_s$ of $V(G)$ such that each $G[U_i]$ has its $k$-subflip-depth bounded by $d$.
For every $G\in\Cc$ fix such a cover.
According to \Cref{lem:sparsify-sfrk} there exist \expos-transductions $\trans {Sparsify}_{k,d}$ and $\trans{Recover}_{k,d}$ such that for every $G_i:=G[U_i]$ there is a reflexive subgraph $G^*_i$ of $G_i$ with tree-depth at most $k\cdot d$ such that 
    \[
        G^*_i \in \trans{Sparsify}_{k,d}(G_i)
        \text{ and } 
        G_i \in \trans{Recover}_{k,d}(G^*_i).
    \]
We define $\trans{Sparsify}$ as the transduction that defines 
\[
G^* := \bigcup_{i=1}^s G_i^*.
\]
This transduction is realized as $\Glue(\trans{Sparsify},s)(G)$ from \cref{lem:gluing-transductions}.
By transduction-closedness of structurally bounded expansion, $\DD$ has structurally bounded expansion.
The graphs in the class $\DD := \{ G^* : G \in \CC \}$ are unions of a bounded number of graphs of bounded tree-depth.
It is well known that bounded tree-depth implies biclique-freeness.
By \Cref{lem:no-bicliques}, also $\DD$ is biclique-free.
By \Cref{cor:sparse-sbe}, $\DD$ has bounded expansion.


\medskip
We now show how to construct $\trans{Recover}$. 
Recall that $G^*=\bigcup_{i=1}^s G_i^*$, where each $G_i^*$ is a (reflexive) subgraph of $G$ on vertex set $U_i$.
It is important to note that in general $G_i^*$ need \emph{not} be an induced subgraph of $G^*$, and in particular not equal to $G^*[U_i]$: $G^*[U_i]$ may contain additional edges coming from some $G_j^*$ with $j\neq i$, in case $U_i$ and $U_j$ overlap. 
Hence, we cannot directly feed $G^*[U_i]$ to $\trans{Recover}_{k,d}$ to recover $G_i$, we must first select the intended edge set of $G_i^*$ as a subgraph of~$G^*$.

Using \Cref{lem:transduce-subgraphs} and the fact that $\DD$ has bounded expansion, there exists an \expos-transduction~$\trans M$ such that
$\trans M(\Dd)$ contains the monotone closure of~$\Dd$, that is, for every $H\in\Dd$, 
\[
    \trans M(H) \supseteq \{H' : H'\subseteq H\text{ is a (reflexive) subgraph}\}.
\]
In particular, for our $G^*\in\Dd$ and every $i\in[s]$ we have $G_i^*\in \trans M(G^*)$.

Let $\trans{Recover}_{k,d}$ be the transduction from \Cref{lem:sparsify-sfrk} (see above). 
For each $i\in[s]$, consider the composition
\[
   \trans R_i \ :=\ \trans{Recover}_{k,d}\circ \trans M.
\]
Then $\trans R_i$ is an \expos-transduction (as a composition of \expos-transductions) and
\[
   G_i \in \trans{Recover}_{k,d}(G_i^*) \subseteq \trans{Recover}_{k,d}(\trans M(G^*)) = \trans R_i(G^*).
\]
Thus, from $G^*$ we can obtain each induced piece $G_i=G[U_i]$ by first using $\trans M$ to guess the correct subgraph~$G_i^*$ and then applying $\trans{Recover}_{k,d}$.
It remains to combine the $s$ recovered graphs $G_1,\dots,G_s$ into~$G$.
Note that $V(G)=U_1\cup\cdots\cup U_s$ and $E(G) =\bigcup_{i=1}^s E(G_i)$, as $U_i$ is a $2$-cover. We therefore define $\trans{Recover}$ as the \expos-transduction that, on input $G^*$ outputs the union of all $\trans R_i$. 

Formally, one can implement this again by a simple product-coloring on the color sets needed for~$\trans M$ and $\trans{Recover}_{k,d}$, indexed by $i\in[s]$, and interpret the final edge predicate as a disjunction over $i\in[s]$ of the edge predicate of the $i$-th recovered copy.
\end{proof}

\Cref{thm:sparsification} now follows by combining \Cref{lem:sbe-sparsify}, \Cref{thm:sbe-and-transclosed}, and \Cref{lem:sparse-notions}.
We will show how this works for the case of bounded clique-width, but the other cases work analogously.

\begin{corollary}
    For every semi-ladder-free class $\CC$ of reflexive graphs of bounded clique-width, there are \expos-transductions $\trans{Sparsify}$ and $\trans{Recover}$ such that every graph $G\in\CC$ contains a reflexive subgraph $G^*$ such that 
    \begin{center}
        $G^* \in \trans{Sparsify}(G)$ and $G \in \trans{Recover}(G^*)$,
    \end{center}
    and the class $\DD := \{G^* : G\in \CC\}$ has bounded tree-width.
\end{corollary}
\begin{proof}
    As $\CC$ is semi-ladder-free and of bounded clique-width, by \Cref{thm:sbe-and-transclosed} it has structurally bounded expansion. By definition of semi-ladder-freeness, $\CC$ is also co-matching-free.
    This means we can apply \Cref{lem:sbe-sparsify} to obtain the transductions $\trans{Sparsify}$ and $\trans{Recover}$, the subgraphs $G^*$, and the class $\DD := \{G^* : G \in \CC\}$.
    The class $\DD$ has bounded clique-width (by transduction-closure of clique-width, \Cref{thm:sbe-and-transclosed}).
    As $\DD$ has bounded expansion, it is biclique-free.
    Now by \Cref{lem:sparse-notions}, $\DD$ has bounded tree-width.
\end{proof}
\section{Positive MSO}\label{sec:positive-mso}

We now prove the collapse of positive MSO to positive FO. 
We use the following standard monotonicity property of positive formulas, see e.g.~\cite[Lem.\ 10.7]{libkin2004elements}.

\begin{lemma}\label{lem:monotonicity}
    Let $\psi(\bar x,Y,\bar X)$ be a formula that is \emph{positive in $Y$}, i.e., every occurrence of $Y$ appears under an even number of negations.
    Then for every graph $G$, every assignment $\bar a, \bar B$ of $\bar x,\bar X$, and all sets $Y\subseteq Y'\subseteq V(G)$,
    \[
    G\models \psi(\bar a,Y,\bar B)\ \Longrightarrow\ G\models \psi(\bar a,Y',\bar B).
    \]
\end{lemma}

\begin{theorem}
Every (existential) positive MSO formula $\phi(\bar x,\bar X)$ is equivalent to an
(existential) positive FO formula.
\end{theorem}

\begin{proof}
We eliminate the set quantifiers one by one. 
First assume $\phi=\exists Y\,\psi(\bar x,Y,\bar X)$ where $\psi$ is positive in~$Y$ (this holds since the whole formula is positive). 
Define $\psi^{\top}(\bar x, \bar X)$ as the formula obtained from $\psi(\bar x,Y, \bar X)$ by replacing every atom $Y(t)$ by $\top$ (equivalently, interpreting $Y$ as the full vertex set). 
We claim that for every $G$, every tuple $\bar a$ and tuple $\bar B$ we have
\[
G \models \exists Y\,\psi(\bar a,Y, \bar B)
\quad\Longleftrightarrow\quad
G \models \psi^{\top}(\bar a, \bar B).
\]
To see this, observe that if $G\models \exists Y\,\psi(\bar a,Y, \bar B)$, then by \cref{lem:monotonicity} also $G\models \psi^{\top}(\bar a, \bar B)$. 
Conversely, if $G\models \psi^{\top}(\bar a,\bar B)$, then choosing $Y=V(G)$ witnesses the
existential quantifier, hence $G\models \exists Y\,\psi(\bar a,Y,\bar B)$.

\medskip
Now assume $\phi=\forall Y\,\psi(\bar x,Y, \bar X)$ with $\psi$ positive in~$Y$. 
Let $\psi^{\bot}(\bar x, \bar X)$ be the formula obtained from $\psi(\bar x,Y, \bar X)$ by replacing every atom $Y(t)$ by $\bot$ (equivalently, interpreting $Y$ as the empty set). Then for every $G,\bar a$, and $\bar B$,
\[
G \models \forall Y\,\psi(\bar a,Y, \bar B)
\quad\Longleftrightarrow\quad
G \models \psi^{\bot}(\bar a, \bar B).
\]
For the forward direction, take $Y=\emptyset$. 
For the reverse direction, if $G\models \psi^{\bot}(\bar a, \bar B)$ then by monotonicity (positivity in~$Y$) we have $G\models \psi(\bar a,Y,\bar B)$ for every $Y\supseteq \emptyset$, i.e., for every set~$Y$, hence $G\models \forall Y\,\psi(\bar a,Y)$.

After eliminating all set quantifiers we obtain a positive FO formula equivalent
to the original positive MSO formula. 
The construction preserves the existential restriction (if present), hence the same argument yields the stated existential version.
\end{proof}

\bibliographystyle{plain}
\bibliography{ref}

@book{Hodges,
  author    = {Wilfrid Hodges},
  title     = {A Shorter Model Theory},
  publisher = {Cambridge University Press},
  year      = {1997}
}

@article{shrubdepth-journal,
  author     = {Robert Ganian and Petr Hlin\v{e}n\'{y} and Jaroslav Nesetril and Jan Obdrz{\'{a}}lek and Patrice Ossona de Mendez},
  bibsource  = {dblp computer science bibliography, https://dblp.org},
  doi        = {10.23638/LMCS-15(1:7)2019},
  journal    = {Log. Methods Comput. Sci.},
  number     = {1},
  title      = {Shrub-depth: Capturing Height of Dense Graphs},
  volume     = {15},
  year       = {2019}
}

@book{nevsetvril2012sparsity,
  author    = {Ne{\v{s}}et{\v{r}}il, Jaroslav and de Mendez, Patrice Ossona},
  publisher = {Springer Science \& Business Media},
  title     = {Sparsity: graphs, structures, and algorithms},
  volume    = {28},
  year      = {2012}
}

@article{dvovrak2018induced,
  author    = {Dvo{\v{r}}{\'a}k, Zden{\v{e}}k},
  journal   = {European Journal of Combinatorics},
  pages     = {143--148},
  publisher = {Elsevier},
  title     = {Induced subdivisions and bounded expansion},
  volume    = {69},
  year      = {2018}
}

@inproceedings{nevsetvril2021rankwidth,
  author       = {Ne{\v{s}}et{\v{r}}il, Jaroslav and Mendez, Patrice Ossona de and Pilipczuk, Micha{\l} and Rabinovich, Roman and Siebertz, Sebastian},
  booktitle    = {Proceedings of the 2021 ACM-SIAM Symposium on Discrete Algorithms (SODA)},
  organization = {SIAM},
  pages        = {2014--2033},
  title        = {Rankwidth meets stability},
  year         = {2021}
}

@inproceedings{pilipczuk2018number,
  title={On the number of types in sparse graphs},
  author={Pilipczuk, Micha{\l} and Siebertz, Sebastian and Toru{\'n}czyk, Szymon},
  booktitle={Proceedings of the 33rd Annual ACM/IEEE Symposium on Logic in Computer Science},
  pages={799--808},
  year={2018}
}

@article{gajarsky2020first,
  author    = {Gajarsk{\'y}, Jakub and Kreutzer, Stephan and Ne{\v{s}}et{\v{r}}il, Jaroslav and Mendez, Patrice Ossona De and Pilipczuk, Micha{\l} and Siebertz, Sebastian and Toru{\'n}czyk, Szymon},
  journal   = {ACM Transactions on Computational Logic (TOCL)},
  number    = {4},
  pages     = {1--41},
  publisher = {ACM New York, NY, USA},
  title     = {First-order interpretations of bounded expansion classes},
  volume    = {21},
  year      = {2020}
}

@inproceedings{nevsetvril2020linear,
  author       = {Ne{\v{s}}et{\v{r}}il, Jaroslav and Rabinovich, Roman and de Mendez, Patrice Ossona and Siebertz, Sebastian},
  booktitle    = {Proceedings of the Fourteenth Annual ACM-SIAM Symposium on Discrete Algorithms},
  organization = {SIAM},
  pages        = {1180--1199},
  title        = {Linear rankwidth meets stability},
  year         = {2020}
}

@article{adler2014interpreting,
  author    = {Adler, Hans and Adler, Isolde},
  journal   = {European Journal of Combinatorics},
  pages     = {322--330},
  publisher = {Elsevier},
  title     = {Interpreting nowhere dense graph classes as a classical notion of model theory},
  volume    = {36},
  year      = {2014}
}

@INPROCEEDINGS {dreier2024stablemc,
author = { Dreier, Jan and Eleftheriadis, Ioannis and M{\"a}hlmann, Nikolas and McCarty, Rose and Pilipczuk, Micha{\l} and Toru{\'n}czyk, Szymon },
booktitle = { 2024 IEEE 65th Annual Symposium on Foundations of Computer Science (FOCS) },
title = {{ First-Order Model Checking on Monadically Stable Graph Classes }},
year = {2024},
volume = {},
ISSN = {},
pages = {21-30},
doi = {10.1109/FOCS61266.2024.00012},
url = {https://doi.ieeecomputersociety.org/10.1109/FOCS61266.2024.00012},
publisher = {IEEE Computer Society},
address = {Los Alamitos, CA, USA},
month =Oct}

@article{baldwin1985second,
  author    = {Baldwin, John T. and Shelah, Saharon},
  journal   = {Notre Dame Journal of Formal Logic},
  number    = {3},
  pages     = {229--303},
  publisher = {Duke University Press},
  title     = {Second-order quantifiers and the complexity of theories.},
  volume    = {26},
  year      = {1985}
}

@book{Diestel,
  author    = {Reinhard Diestel},
  publisher = {Springer},
  series    = {Graduate texts in mathematics},
  title     = {Graph Theory, 4th Edition},
  volume    = {173},
  year      = {2012}
}

@article{stable_graphs,
  author     = {Podewski, Klaus-Peter and Ziegler, Martin},
  journal    = {Fundamenta Mathematicae},
  number     = {2},
  pages      = {101-107},
  title      = {Stable graphs},
  url        = {http://eudml.org/doc/210953},
  volume     = {100},
  year       = {1978}
}

@inproceedings{fabianski2019,
  address    = {Dagstuhl, Germany},
  author     = {Grzegorz Fabianski and Micha{\l} Pilipczuk and Sebastian Siebertz and Szymon Toru\'{n}czyk},
  booktitle  = {36th International Symposium on Theoretical Aspects of Computer Science (STACS 2019)},
  doi        = {10.4230/LIPIcs.STACS.2019.27},
  editor     = {Rolf Niedermeier and Christophe Paul},
  isbn       = {978-3-95977-100-9},
  issn       = {1868-8969},
  pages      = {27:1--27:16},
  publisher  = {Schloss Dagstuhl--Leibniz-Zentrum fuer Informatik},
  series     = {Leibniz International Proceedings in Informatics (LIPIcs)},
  title      = {{Progressive Algorithms for Domination and Independence}},
  url        = {http://drops.dagstuhl.de/opus/volltexte/2019/10266},
  volume     = {126},
  year       = {2019}
}

@inproceedings{flipper-game,
  address    = {Dagstuhl, Germany},
  author     = {Gajarsk\'{y}, Jakub and M\"{a}hlmann, Nikolas and McCarty, Rose and Ohlmann, Pierre and Pilipczuk, Micha{\l} and Przybyszewski, Wojciech and Siebertz, Sebastian and Soko{\l}owski, Marek and Toru\'{n}czyk, Szymon},
  booktitle  = {50th International Colloquium on Automata, Languages, and Programming (ICALP 2023)},
  doi        = {10.4230/LIPIcs.ICALP.2023.128},
  editor     = {Etessami, Kousha and Feige, Uriel and Puppis, Gabriele},
  isbn       = {978-3-95977-278-5},
  issn       = {1868-8969},
  pages      = {128:1--128:16},
  publisher  = {Schloss Dagstuhl -- Leibniz-Zentrum f{\"u}r Informatik},
  series     = {Leibniz International Proceedings in Informatics (LIPIcs)},
  title      = {{Flipper Games for Monadically Stable Graph Classes}},
  url        = {https://drops.dagstuhl.de/entities/document/10.4230/LIPIcs.ICALP.2023.128},
  urn        = {urn:nbn:de:0030-drops-181804},
  volume     = {261},
  year       = {2023}
}

@inproceedings{dreier2023ssmc,
  address    = {New York, NY, USA},
  author     = {Dreier, Jan and M\"{a}hlmann, Nikolas and Siebertz, Sebastian},
  booktitle  = {Proceedings of the 55th Annual ACM Symposium on Theory of Computing},
  doi        = {10.1145/3564246.3585186},
  isbn       = {9781450399135},
  location   = {Orlando, FL, USA},
  numpages   = {14},
  pages      = {567-580},
  publisher  = {Association for Computing Machinery},
  series     = {STOC 2023},
  title      = {First-Order Model Checking on Structurally Sparse Graph Classes},
  url        = {https://doi.org/10.1145/3564246.3585186},
  year       = {2023}
}

@inproceedings{dreier2022indiscernibles,
  author     = {Jan Dreier and Nikolas M{\"{a}}hlmann and Sebastian Siebertz and Szymon Toru{\'n}czyk},
  bibsource  = {dblp computer science bibliography, https://dblp.org},
  booktitle  = {50th International Colloquium on Automata, Languages, and Programming, {ICALP} 2023, July 10-14, 2023, Paderborn, Germany},
  doi        = {10.4230/LIPIcs.ICALP.2023.125},
  editor     = {Kousha Etessami and Uriel Feige and Gabriele Puppis},
  pages      = {125:1--125:18},
  publisher  = {Schloss Dagstuhl - Leibniz-Zentrum f{\"{u}}r Informatik},
  series     = {LIPIcs},
  title      = {Indiscernibles and Flatness in Monadically Stable and Monadically {NIP} Classes},
  url        = {https://doi.org/10.4230/LIPIcs.ICALP.2023.125},
  volume     = {261},
  year       = {2023}
}

@article{twwI,
    author = {Bonnet, \'{E}douard and Kim, Eun Jung and Thomass\'{e}, St\'{e}phan and Watrigant, R\'{e}mi},
    title = {Twin-width I: Tractable FO Model Checking},
    year = {2021},
    issue_date = {February 2022},
    publisher = {Association for Computing Machinery},
    address = {New York, NY, USA},
    volume = {69},
    number = {1},
    issn = {0004-5411},
    url = {https://doi.org/10.1145/3486655},
    doi = {10.1145/3486655},
    journal = {J. ACM},
    month = {nov},
    articleno = {3},
    numpages = {46}
    }

@article{twwII,
  author     = {Bonnet, {\'E}douard and Geniet, Colin and Kim, Eun Jung and Thomass{\'e}, St{\'e}phan and Watrigant, R{\'e}mi},
  doi        = {10.5070/c62257876},
  journal    = {Combinatorial Theory},
  number     = {2},
  title      = {{Twin-width II: Small classes}},
  volume     = {2},
  year       = {2022}
}

@inproceedings{shrubdepth,
  address   = {Berlin, Heidelberg},
  author    = {Ganian, Robert and Hlin{\v{e}}n{\'y}, Petr and Ne{\v{s}}et{\v{r}}il, Jaroslav and Obdr{\v{z}}{\'a}lek, Jan and Ossona de Mendez, Patrice and Ramadurai, Reshma},
  booktitle = {Mathematical Foundations of Computer Science 2012},
  editor    = {Rovan, Branislav and Sassone, Vladimiro and Widmayer, Peter},
  isbn      = {978-3-642-32589-2},
  pages     = {419--430},
  publisher = {Springer Berlin Heidelberg},
  title     = {When Trees Grow Low: Shrubs and Fast {MSO1}},
  year      = {2012}
}

@inproceedings{POM21,
  author     = {Patrice {Ossona de Mendez}},
  bibsource  = {dblp computer science bibliography, https://dblp.org},
  booktitle  = {38th International Symposium on Theoretical Aspects of Computer Science, {STACS} 2021},
  doi        = {10.4230/LIPIcs.STACS.2021.2},
  pages      = {2:1--2:7},
  publisher  = {Schloss Dagstuhl --- Leibniz-Zentrum f{\"{u}}r Informatik},
  series     = {LIPIcs},
  title      = {First-Order Transductions of Graphs (Invited Talk)},
  url        = {https://doi.org/10.4230/LIPIcs.STACS.2021.2},
  volume     = {187},
  year       = {2021}
}

@book{libkin2004elements,
  title     = {Elements of finite model theory},
  author    = {Libkin, Leonid},
  volume    = {41},
  year      = {2004},
  publisher = {Springer}
}

@inproceedings{dreier2024flipbreakability,
author = {Dreier, Jan and M\"{a}hlmann, Nikolas and Toru\'{n}czyk, Szymon},
title = {Flip-Breakability: A Combinatorial Dichotomy for Monadically Dependent Graph Classes},
year = {2024},
isbn = {9798400703836},
publisher = {Association for Computing Machinery},
address = {New York, NY, USA},
url = {https://doi.org/10.1145/3618260.3649739},
doi = {10.1145/3618260.3649739},
booktitle = {Proceedings of the 56th Annual ACM Symposium on Theory of Computing},
pages = {1550–1560},
numpages = {11},
location = {Vancouver, BC, Canada},
series = {STOC 2024}
}

@inproceedings{gajarsky2022stable,
  author    = {Gajarsk{\'y}, Jakub and Pilipczuk, Micha{\l} and Toru{\'n}czyk, Szymon},
  booktitle = {Proceedings of the 37th Annual ACM/IEEE Symposium on Logic in Computer Science},
  pages     = {1--12},
  title     = {Stable graphs of bounded twin-width},
  year      = {2022}
}

@article{nesetril2020structuralc,
  title={Structural properties of the first-order transduction quasiorder},
  author={Nesetril, Jaroslav and de Mendez, Patrice Ossona and Siebertz, Sebastian},
  journal={arXiv preprint arXiv:2010.02607},
  year={2020}
}

@inproceedings{gurski2000tree,
  title={The tree-width of clique-width bounded graphs without K n, n},
  author={Gurski, Frank and Wanke, Egon},
  booktitle={International Workshop on Graph-Theoretic Concepts in Computer Science},
  pages={196--205},
  year={2000},
  organization={Springer}
}

@inproceedings{dreier2022combinatorial,
  author    = {Dreier, Jan and M\"{a}hlmann, Nikolas and Mouawad, Amer E. and Siebertz, Sebastian and Vigny, Alexandre},
  title     = {{Combinatorial and Algorithmic Aspects of Monadic Stability}},
  booktitle = {33rd International Symposium on Algorithms and Computation (ISAAC 2022)},
  pages     = {11:1--11:17},
  series    = {Leibniz International Proceedings in Informatics (LIPIcs)},
  isbn      = {978-3-95977-258-7},
  issn      = {1868-8969},
  year      = {2022},
  volume    = {248},
  editor    = {Bae, Sang Won and Park, Heejin},
  publisher = {Schloss Dagstuhl -- Leibniz-Zentrum f{\"u}r Informatik},
  address   = {Dagstuhl, Germany},
  url       = {https://drops.dagstuhl.de/entities/document/10.4230/LIPIcs.ISAAC.2022.11},
  urn       = {urn:nbn:de:0030-drops-172961},
  doi       = {10.4230/LIPIcs.ISAAC.2022.11}
}

@article{bounded_expansion,
title = {Grad and classes with bounded expansion I. Decompositions},
journal = {European Journal of Combinatorics},
volume = {29},
number = {3},
pages = {760-776},
year = {2008},
issn = {0195-6698},
doi = {https://doi.org/10.1016/j.ejc.2006.07.013},
url = {https://www.sciencedirect.com/science/article/pii/S019566980700056X},
author = {Jaroslav Nešetřil and Patrice {Ossona de Mendez}}
}

@phdthesis{maehlmann-thesis,
  author = {Nikolas M\"{a}hlmann},
  school = {University of Bremen},
  title  = {Monadically Stable and Monadically Dependent Graph Classes: Characterizations and Algorithmic Meta-Theorems},
  year   = {2024}
}

@InProceedings{flip-separability,
  author =	{Bonnet, \'{E}douard and Braunfeld, Samuel and Eleftheriadis, Ioannis and Geniet, Colin and M\"{a}hlmann, Nikolas and Pilipczuk, Micha{\l} and Przybyszewski, Wojciech and Toru\'{n}czyk, Szymon},
  title =	{{Separability Properties of Monadically Dependent Graph Classes}},
  booktitle =	{52nd International Colloquium on Automata, Languages, and Programming (ICALP 2025)},
  pages =	{147:1--147:19},
  series =	{Leibniz International Proceedings in Informatics (LIPIcs)},
  ISBN =	{978-3-95977-372-0},
  ISSN =	{1868-8969},
  year =	{2025},
  volume =	{334},
  editor =	{Censor-Hillel, Keren and Grandoni, Fabrizio and Ouaknine, Jo\"{e}l and Puppis, Gabriele},
  publisher =	{Schloss Dagstuhl -- Leibniz-Zentrum f{\"u}r Informatik},
  address =	{Dagstuhl, Germany},
  URL =		{https://drops.dagstuhl.de/entities/document/10.4230/LIPIcs.ICALP.2025.147},
  URN =		{urn:nbn:de:0030-drops-235246},
  doi =		{10.4230/LIPIcs.ICALP.2025.147}
}

@misc{pilipczuk2025graphclasseslenslogic,
      title={Graph classes through the lens of logic}, 
      author={Michał Pilipczuk},
      year={2025},
      eprint={2501.04166},
      archivePrefix={arXiv},
      primaryClass={math.CO},
      url={https://arxiv.org/abs/2501.04166}, 
}

@article{courcelle2000upper,
title = {Upper bounds to the clique width of graphs},
journal = {Discrete Applied Mathematics},
volume = {101},
number = {1},
pages = {77-114},
year = {2000},
issn = {0166-218X},
doi = {https://doi.org/10.1016/S0166-218X(99)00184-5},
url = {https://www.sciencedirect.com/science/article/pii/S0166218X99001845},
author = {Bruno Courcelle and Stephan Olariu}
}

@inproceedings{merge-width,
author = {Dreier, Jan and Toru\'{n}czyk, Szymon},
title = {Merge-Width and First-Order Model Checking},
year = {2025},
isbn = {9798400715105},
publisher = {Association for Computing Machinery},
address = {New York, NY, USA},
url = {https://doi.org/10.1145/3717823.3718259},
doi = {10.1145/3717823.3718259},
booktitle = {Proceedings of the 57th Annual ACM Symposium on Theory of Computing},
pages = {1944–1955},
numpages = {12},
location = {Prague, Czechia},
series = {STOC '25}
}

@article{dreier2026efficient,
  author = {Dreier, Jan and Gajarsk\'{y}, Jakub and Pilipczuk, Micha{\l}},
  title = {Efficient reversal of transductions of sparse graph classes},
  year = {2026},
  eprint = {2601.14906},
  archiveprefix = {arXiv},
  journal = {arXiv preprint arXiv:2601.14906}
}

\end{document}